\documentclass[a4paper,11pt]{article}
\usepackage{jheppub} 
\usepackage{lineno}
\usepackage{amsmath}
\usepackage{amssymb}
\usepackage{amsfonts}
\usepackage{amssymb}
\usepackage{amsthm}
\usepackage{amscd}
\usepackage{xcolor}

\newcommand{\ew}{\mathcal{E}}%
\newcommand{\cw}{\mathcal{C}}%
\newcommand{\bb}{\partial M}%
\newcommand{\mm}{\mathcal{M}}%
\newcommand{\mSE}{\tilde{\mathcal{S}}_E}%
\newcommand{\NN}{\mathcal{N}}%

\newcounter{Counter}
\newtheorem{definition}{Definition}[section]
\newtheorem{lemma}[definition]{Lemma}
\newtheorem{conjecture}[definition]{Conjecture}
\newtheorem{theorem}[definition]{Theorem}

\newtheorem{corollary}[definition]{Corollary}
\newtheorem{example}[definition]{Example}
\newtheorem{remark}[definition]{Remark}
\newtheorem{assumption}{Assumption}

\arxivnumber{2509.23119 } 

\title{\boldmath A Proof of the Generalized Connected Wedge Theorem}

\author{Bowen Zhao}
\affiliation{Beijing Institute of Mathematical Sciences and Applications, Beijing, China}

\emailAdd{bowenzhao@bimsa.cn}

\abstract{
In the context of asymptotic $2$-to-$2$ scattering process in AdS/CFT, the Connected Wedge Theorem identifies the existence of $O(1/G_N)$ mutual information between suitable boundary subregions, referred to as decision regions, as a necessary but not sufficient condition for bulk-only scattering processes, i.e., nonempty bulk scattering region $\mathcal{S}_0$. Recently, Liu and Leutheusser proposed an enlarged bulk scattering region $\mathcal{S}_E$ and conjectured that the non-emptiness of $\mathcal{S}_E$ fully characterizes the existence of $O(1/G_N)$ mutual information between decision regions. Here, we provide a geometrical or general relativity proof for a slightly modified version of their conjecture.}


\begin{document}
\maketitle
\flushbottom

\section{Introduction}\label{sec:intro}
The AdS/CFT correspondence posits a duality between a quantum gravity theory in an asymptotically Anti-de Sitter (AdS) spacetime $M$ and a conformal field theory (CFT) on its timelike boundary $\bb$ \cite{maldacena1999AdSCFT,witten1998AdScft}. A foundational requirement for this duality is the consistency of causal structure between the bulk and the boundary.

\subsection{Two-Point Causality Constraints and the Gao-Wald Theorem}

The first layer of this consistency is encoded in the theorem of Gao and Wald \cite{gao2000theorems}. It states that bulk causality cannot violate boundary causality. Specifically, assuming that $M$ satisfies the null curvature condition\footnote{This is equivalent to the null energy condition $T_{ab}k^ak^b \geq 0$ when the Einstein equation is imposed.}
\begin{equation}\label{eq:NEC}
    R_{ab}k^a k^b \geq 0,  \quad \text{ for all null } k^a
\end{equation}
and that the conformal completion $\overline{M}$ is globally hyperbolic\footnote{The original statement in \cite{gao2000theorems} assumes that $\overline{M}$ is strongly causal and that all causal diamonds in $\overline{M}$ are compact. As discussed in Chapter 8 of ref. \cite{waldGR}, this set of conditions is equivalent to the definition of global hyperbolicity based on the existence of a Cauchy surface. We also refer to this condition as $M$ is AdS-hyperbolic.},
then any two boundary points $p,q\in \bb$ that can be connected by a causal curve through the bulk must also be connectable by a causal curve restricted entirely to the boundary $\bb$. A stronger result follows upon also assuming the null generic condition that every null geodesic with tangent $k^a$ encounters a point where
\begin{equation}\label{eq:NGC}
k_{[a}R_{b]cd[e}k_{f]}k^c k^d \neq 0.
\end{equation}
This condition ensures null geodesics experience non-trivial curvature, thereby excluding pure AdS. In this context, any \emph{prompt} \footnote{A null geodesic $l$ is called prompt if for every pair of points $p, q \in l$, either $p \in \partial J^+[q]$ or $q \in \partial J^+[p]$. This term captures that no causal curve connects $p$ and $q$ more quickly than the segment of $l$ itself, making $l$ an optimal signaling path. For a detailed discussion, see \cite{witten2020light}.} causal curve between two boundary points must lie entirely on the boundary $\bb$ \cite{witten2020light}.

\subsection{The Connected Wedge Theorem: Beyond Two-Point Causality}
A more profound consistency requirement emerges when considering asymptotic quantum tasks involving multiple regions on $\partial M$. 
In $(d+1)$-dimensional asymptotically AdS spacetime $M$, asymptotic $2$-to-$d$ scattering configurations with $2$ input points and $d$ output points on $\partial M$ may permit local scattering processes in the bulk that have no counterpart in the boundary theory. Following established terminology, we refer to such configurations as \emph{bulk-only scattering configurations}.

Such configurations were initially investigated in refs. \cite{gary2009bulkonly,heemskerk2009bulkonly,penedones2011bulkonly,maldacena2017bulkpointsingularity} through the study of perturbative singularities in Lorentzian correlators of boundary local operators. May \cite{may2019quantumtask} re-examined these configurations from a quantum information perspective and conjectured that the existence of bulk-only scattering configurations necessarily implies substantial mutual information between specific boundary regions.

\begin{figure}
    \centering
    \includegraphics[width=0.6\linewidth,trim={3cm 12cm 2cm 5cm},clip]{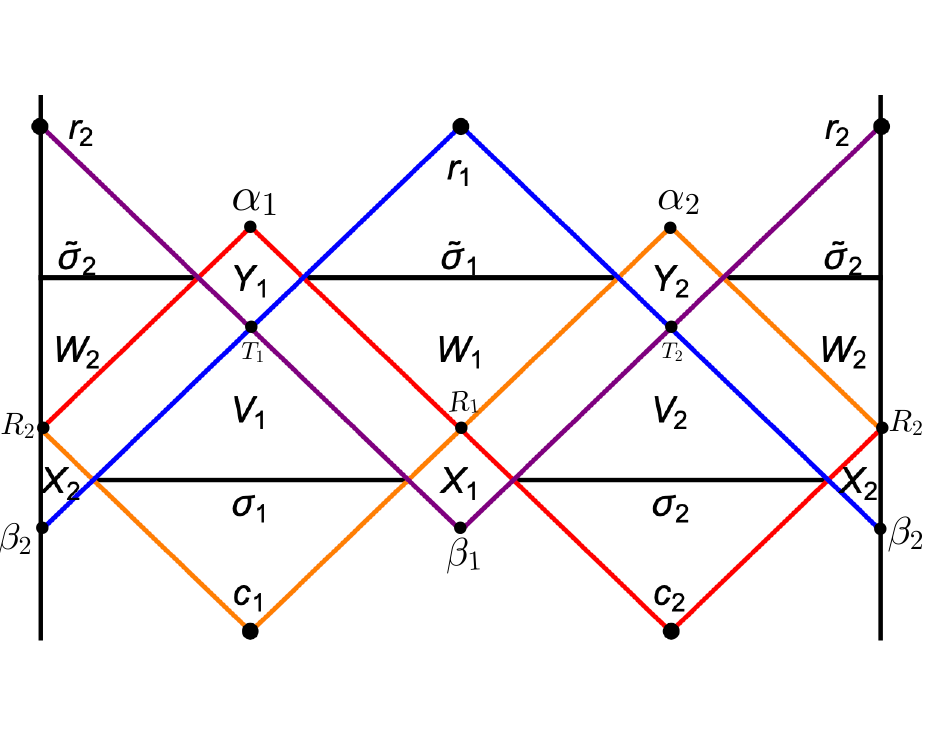}
    \caption{Boundary set-up of $2$-to-$2$ scattering process. This figure is adjusted from \cite{LL2025superadditivity}.}
    \label{fig:setup}
\end{figure}

The boundary setup for an asymptotic $2$-to-$2$ scattering configuration is illustrated in Figure \ref{fig:setup}. Inputs, which can be classical bits or quantum bits (qubits), are received at points $c_1$ and $c_2$, while outputs are required at both $r_1$ and $r_2$.
The input/decision regions $V_1$ and $V_2$, and output regions $W_1$ and $W_2$, are defined as follows:
\begin{align}\label{eq:V_W}
V_1 &= \hat{J}^+[c_1] \cap \hat{J}^-[r_1] \cap \hat{J}^-[r_2], \nonumber \\
V_2 &= \hat{J}^+[c_2] \cap \hat{J}^-[r_1] \cap \hat{J}^-[r_2], \nonumber \\
W_1 &= \hat{J}^-[r_1] \cap \hat{J}^+[c_1] \cap \hat{J}^+[c_2], \nonumber \\
W_2 &= \hat{J}^-[r_2] \cap \hat{J}^+[c_1] \cap \hat{J}^+[c_2],
\end{align}
where $\hat{J}^\pm(c_1)$ denotes the causal future/past of point $c_1$ on the boundary $\partial M$, and similarly for other points.

For bulk-only scattering configurations, the boundary scattering region
\begin{equation}\label{eq:Shat}
\hat{\mathcal{S}} = \hat{J}^+[c_1] \cap \hat{J}^+[c_2] \cap \hat{J}^-[r_1] \cap \hat{J}^-[r_2]
\end{equation}
is empty, implying that the intersection of any two of the four regions $V_1, V_2, W_1, W_2$ must be empty.
Let ${J}^\pm[c_1]$ denote the causal future/past of point $c_1$ in the bulk $M$. The bulk scattering region
\begin{equation}\label{eq:S0}
\mathcal{S}_0 = J^+[c_1] \cap J^+[c_2] \cap J^-[r_1] \cap J^-[r_2]
\end{equation}
is, by definition, non-empty for bulk-only scattering configurations. 

The Connected Wedge Theorem (CWT) \cite{maypenington2020holographic} establishes that for $2$-to-$2$ bulk-only scattering in AdS$_3$, the boundary regions $V_1$ and $V_2$ must share $O(1/G_N)$ mutual information:
 $$I(V_1:V_2) = S(V_1) + S(V_2) -S(V_1\cup V_2) \sim O(1/G_N)$$
This implies that a local bulk scattering process without a boundary counterpart necessitates nonlocal boundary protocols.

The large mutual information admits a direct geometric interpretation through the Hubeny-Rangamani-Ryu-Takayanagi (HRRT) formula, $S(A) = \min_{\gamma \sim A} \left(\frac{|\gamma|}{4G_N}\right) + O(1)$. The $O(1/G_N)$ mutual information indicates that the extremal surface computing $S(V_1\cup V_2)$ is not the union of the individual extremal surfaces for $V_1$ and $V_2$. Consequently, the entanglement wedge $\ew(V_1\cup V_2)$ becomes connected. This leads to the geometric formulation of the CWT:

\begin{theorem}\label{thm:CWT}
Assume that the bulk spacetime $M$ is AdS-hyperbolic\footnote{Recall that this means the conformal completion $\overline{M}$ is globally hyperbolic.} and satisfies the null energy condition \eqref{eq:NEC}, and that the HRRT surface can be found via a maximin procedure \cite{wall2014maximin}. Then $\mathcal{S}_0 \neq \emptyset$ implies that the entanglement wedge of $V_1 \cup V_2$ is connected.
\end{theorem}

The converse, however, does not hold universally; a connected entanglement wedge does not guarantee a non-empty bulk scattering region, as demonstrated by an explicit counterexample in \cite{maypenington2020holographic}.

\subsection{Algebraic Reformulation and the Generalized Connected Wedge Theorem}
A recent work by Liu and Leutheusser \cite{LL2025superadditivity} reinterprets bulk-only scattering configurations in terms of operator algebras, leading to a conjectured generalization of the CWT.

They identify emergent superadditivity in the large-$N$ boundary algebra as the key feature enabling non-local boundary protocols, contrasting with the additive nature of additive finite-$N$ algebra. While superadditivity in standard quantum error-correcting codes (involving finite-dimensional Type I algebras) refers to properties of encoding maps, the large-$N$ boundary algebras in AdS/CFT are Type III$_1$ factors where superadditivity becomes an intrinsic algebraic property. Although the boundary CFT functions as an error-correcting code for the bulk \cite{almheiri2015bulk}, finite-dimensional models appear insufficient to fully capture entanglement wedge reconstruction (see Section III of \cite{LL2025superadditivity}).

For von Neumann algebras $\mm_{V_1}$ and $\mm_{V_2}$ of boundary regions $V_1$ and $V_2$
\begin{itemize}
    \item Additivity would mean $\mm_{V_1} \vee \mm_{V_2} = (\mm_{V_1} \cup \mm_{V_2})'' =\mm_{V_1\cup V_2} $
    \item Superadditivity states $\mm_{V_1} \vee \mm_{V_2} = (\mm_{V_1} \cup \mm_{V_2})'' \subseteq \mm_{V_1\cup V_2} $
\end{itemize}
This algebraic property finds a geometric dual through entanglement wedge reconstruction $\mm_{V} = \mm^{(\text{bulk})}_{\ew(V)}$, where $\mm^{(\text{bulk})}$ indicates bulk operator algebra and $\ew(V)$ is the entanglement wedge of $V$.
Assuming additivity holds in the bulk, i.e., $\mm^{(\text{bulk})}_{\ew(V_1)}\vee \mm^{(\text{bulk})}_{\ew(V_2)}=\mm^{(\text{bulk})}_{\ew(V_1)\cup\ew(V_2)}$, superadditivity of the boundary algebra is equivalent to an inclusion of entanglement wedges:
 $$\ew(V_1)\cup \ew(V_2) \subseteq \ew(V_1\cup V_2).$$

This inclusion is the algebraic signature of the geometric phase transition in which the HRRT surface jumps and the entanglement wedge becomes connected. Liu and Leutheusser identify additional operators in $\mathcal{M}_{\ew(V_1 \cup V_2)}$ — those not contained in the algebra generated by $\mathcal{M}_{\ew(V_1)}$ and $\mathcal{M}_{\ew(V_2)}$ as the fundamental resources that facilitate the non-local boundary protocols required to simulate bulk-local tasks. Their conjecture is as follows:
\begin{conjecture}
Define the generalized scattering region
\begin{equation}
\mathcal{S}_E = \ew(V_1 \cup V_2) \cap \ew(W_1 \cup W_2).
\end{equation}
Under assumptions analogous to those in Theorem \ref{thm:CWT}\footnote{Precise assumptions are specified in Section \ref{subsec:notation_assumptions}.}, $\mathcal{S}_E\neq \emptyset$ if and only if $\ew(V_1 \cup V_2)$ is connected.
\end{conjecture}

We proved a modified version of the above conjecture, establishing that connectedness of both $\ew(V_1 \cup V_2)$ and $\ew(W_1 \cup W_2)$ implies $\mathcal{S}_E \neq \emptyset$, thereby identifying the geometric resource required for quantum tasks. Our work provided the first proof of this non-trivial direction.

An independent study by \cite{lima2025sufficientGCWT} subsequently introduced the alternative bulk region 
\begin{equation}
    \tilde{\mathcal{S}}_E=\tilde{S}_E =[\ew(V_1\cup V_2) \setminus (\ew(V_1)\cup \ew(V_2))] \, \bigcap \, [\ew(W_1\cup W_2) \setminus (\ew(W_1)\cup \ew(W_2))]
\end{equation}
and established the full equivalence: non-emptiness of $\tilde{\mathcal{S}}_E$ occurs if and only if both entanglement wedges are connected, with the converse direction following directly from the region's construction.

In this updated version, we demonstrate that both $\mathcal{S}_E$ and $\tilde{\mathcal{S}}_E$ are nonempty when the entanglement wedges are connected. Our analysis provides a comprehensive geometric understanding of these scattering regions through rigorous causal and entanglement wedge arguments. A key contribution is the identification of the causal anchoring principle (Section \ref{subsec:proof_strategy}) as a fundamental mechanism in asymptotic scattering problems, which underlies our proof strategy and offers clear geometric insight.

Our main results are Theorems \ref{thm:SE_wedge} and \ref{thm:GCWT}, which we summarize as follows:
\begin{theorem}
Assume conditions analogous to those in Theorem \ref{thm:CWT} \footnote{Precise assumptions are specified in Section \ref{subsec:notation_assumptions}.}.
\begin{itemize}
\item If both $\ew(V_1\cup V_2)$ and $\ew(W_1\cup W_2)$ are connected then $\mathcal{S}_E \neq \emptyset$.
\item The region $\tilde{\mathcal{S}}_E$ is nonempty if and only if $\ew(V_1 \cup V_2)$ and $\ew(W_1 \cup W_2)$ are both connected.
\end{itemize}
\end{theorem}

As noted in \cite{maypenington2020holographic}, bulk-only $2$-to-$2$ scattering is generally absent in dimensions greater than three for an asymptotically global AdS spacetime but could exist in asymptotically locally AdS spacetimes (the asymptotic boundary is topologically different from that of $AdS$, e.g. AdS soliton\cite{horowitzmyers}). 
We here focus on $2$-to-$2$ scattering in asymptotically $\text{AdS}_3$ spacetimes, which are dual to CFTs on the cylindrical boundary $\mathbb{R} \times S^1$ and leave discussion about $n$-to-$n$ ($n>2$) and higher dimensions to future work. 

We lastly mention that while the focus here is to enlarge the bulk scattering region $\mathcal{S}_0$ to achieve an equivalence statement for existence of $(1/G_N)$ mutual information on the boundary, other studies \cite{caminiti2025geodesics} attempt to find a stricter condition than the entanglement wedges being connected to achieve an equivalence statement for non-emptiness of $\mathcal{S}_0$.

The rest of this paper is organized as follows. In Section \ref{subsec:review}, we review geometric features of the $2$-to-$2$ scattering setup. In section \ref{subsec:proof_strategy} we establish the causal anchoring principle in asymptotic scattering problems and state our proof strategy. In Section \ref{subsec:S0_SE_inc}, we prove the inclusion $\mathcal{S}_0 \subseteq \mathcal{S}_E$ to establish our proof strategy. In Section \ref{subsec:GCWT}, we first show that connectedness of the entanglement wedges implies $\mathcal{S}_E \neq \emptyset$, and then observe that the proof can be easily extended to show $\tilde{\mathcal{S}}_E \neq \emptyset$. We then conclude with the equivalence between $\tilde{\mathcal{S}}_E \neq \emptyset$ and connectedness of entanglement wedges and some remarks about the relation between $\tilde{\mathcal{S}}_E $ and $\mathcal{S}_E$.  We discuss implications and future directions in Section \ref{section:discussion}. The appendix contains some basic facts that are used in the main text.

\subsection{Notations and Assumptions}\label{subsec:notation_assumptions}
Here we summarize the notations, conventions, and assumptions used throughout this paper.

We adopt natural units with $\hbar=c=1$ and set the AdS length scale $l_{\text{AdS}}=1$, while keeping Newton's constant $G_N$ explicit. Our notation follows ref. \cite{waldGR}, using the mostly-plus metric signature.

\begin{itemize}
\item \textbf{Spacetime regions:} Bulk regions are denoted by script letters ($\mathcal{U}, \mathcal{V}, \mathcal{W}, \cdots$), while boundary regions use straight capitals ($U, V, W, \cdots$). The same symbol may denote either a causal diamond or its Cauchy surface, with the meaning clear from context.

\item \textbf{Cauchy slices:} Bulk Cauchy slices are denoted by $\Sigma$ with appropriate subscripts, boundary Cauchy slices by $\hat{\Sigma}$ with subscripts. By abuse of notation, $\Sigma$ may also refer to Cauchy slices of the conformally compactified spacetime.

\item \textbf{Causal structure:} The bulk causal future/past of region $\mathcal{V}$ is $J^\pm[\mathcal{V}]$; for boundary region $V$, we write $J^\pm[V]$ for bulk causal influence and $\hat{J}^\pm[V]$ for boundary causal influence.

\item \textbf{Domains of dependence:} The bulk domain of dependence of $\mathcal{V}$ is $\mathcal{D}[\mathcal{V}]$; the boundary domain of dependence of $V$ is $\hat{D}[V]$.

\item \textbf{Entanglement structures:} For boundary region $V$, we denote the entanglement wedge by $\ew(V)$, causal wedge by $\cw(V)$, and HRRT surface by $\text{RT}(V)$.

\item \textbf{Complements:} The causal complement (bulk or boundary) uses prime notation ($'$), while set-theoretic complement within a Cauchy slice uses superscript $c$.
\end{itemize}

\begin{assumption}\label{assumption:1}
We assume throughout that:
\begin{enumerate}
\item The bulk spacetime $M$ satisfies the null curvature condition \eqref{eq:NEC};
\item HRRT surfaces can be found via a maximin procedure;
\item The spacetime is AdS-hyperbolic (the conformal compactification $\overline{M}=M \cup \partial M$ admits a Cauchy slice);
\item The spacetime region between some Cauchy slice preceding $\ew(V_1\cup V_2)$ and some Cauchy slice following $\ew(W_1\cup W_2)$ is singularity-free.
\item The global boundary state is pure, ensuring that a boundary region $V$ and its causal complement $V'$ share the same HRRT surface.
\end{enumerate}
\end{assumption}

\section{Geometric Proof Using General Relativity}
\subsection{Geometric Structure of 2-to-2 Bulk-Only Scattering}\label{subsec:review}
We start by reviewing some facts about $2$-to-$2$ bulk-only scattering configurations in $2+1$ dimension and collecting some relevant theorems.

Recall that the input/decision regions $V_1, V_2$ and output regions $W_1, W_2$, as defined in \eqref{eq:V_W}, are all non-empty sets on $\bb$. Since the intersection of any two of the four regions would imply a non-empty boundary scattering region $\hat{S}$, defined in \eqref{eq:Shat}, we restrict our analysis to the non-trivial case where no two such regions intersect. 
This condition implies that the causal complement of $V_1 \cup V_2$ (and similarly of $W_1 \cup W_2$) is non-empty and, crucially in $2+1$ dimensions, necessarily disconnected.
We label the connected components of the causal complement of $V_1 \cup V_2$ as $X_1$ and $X_2$, and those of $W_1 \cup W_2$ as $Y_1$ and $Y_2$, with each $X_i$ lying in the causal past of $r_i$ and each $Y_i$ in the causal future of $c_i$, as illustrated in Figure \ref{fig:setup}. For later use, we note that each $Y_i$ (for $i=1,2$) lies to the future of the corresponding $V_i$, with the two sharing a boundary point labeled $T_i$ in Figure \ref{fig:setup}. Similarly, each $X_i$ lies to the past of the corresponding $W_i$, with the two sharing a boundary point $R_i$ (Figure \ref{fig:setup}).

We denote by $\alpha_1$ the future antipodal point of $c_2$ on $\partial M$, and by $\alpha_2$ the future antipodal point of $c_1$ on $\partial M$. In other words, the two future-directed null geodesics emanating from $c_1$ converge at $\alpha_2$, while those from $c_2$ converge at $\alpha_1$.
Similarly, we denote by $\beta_1$ the past antipodal point of $r_2$ on $\partial M$, and by $\beta_2$ the past antipodal point of $r_1$ on $\partial M$. That is, the two past-directed null geodesics from $r_1$ converge at $\beta_2$, while those from $r_2$ converge at $\beta_1$.

Since $X_i$ and $Y_j$ are defined as the causal complements of $V_i$ and $W_j$, respectively, we obtain (see Figure \ref{fig:setup}):
\begin{align}\label{eq:alpha_beta}
    X_1&=\hat{J}^+[\beta_1] \cap \hat{J}^-[\alpha_1] \cap \hat{J}^-[\alpha_2] , \nonumber \\
    X_2&=\hat{J}^+[\beta_2] \cap \hat{J}^-[\alpha_1] \cap \hat{J}^-[\alpha_2], \nonumber  \\
    Y_1&=\hat{J}^-[\alpha_1] \cap \hat{J}^+[\beta_1] \cap \hat{J}^-[\beta_2], \nonumber  \\
    Y_2&=\hat{J}^-[\alpha_2] \cap \hat{J}^+[\beta_1] \cap \hat{J}^-[\beta_2] .
\end{align}
We will make extensive use of the following key boundary relations:
\begin{align}
\partial \hat{J}^+[\beta_1] = \partial \hat{J}^-[r_2], \quad
\partial \hat{J}^+[\beta_2] = \partial \hat{J}^-[r_1], \quad
\partial \hat{J}^-[\alpha_1] = \partial \hat{J}^+[c_2], \quad
\partial \hat{J}^-[\alpha_2] = \partial \hat{J}^+[c_1].
\end{align}
\label{eq:relation_X_Y_c_r}
In contrast, the bulk analog $\partial J^+[\beta_1] = \partial J^-[r_2]$ generally fails. Matter or curvature in the bulk retards causal curves relative to pure AdS, causing the null hypersurfaces $\partial J^+[\beta_1]$ and $\partial J^-[r_2]$ to separate—specifically, $\partial J^-[r_2]$ lies to the past of $\partial J^+[\beta_1]$. Consequently, future-directed bulk causal curves from $\beta_1$ reach the boundary only at points strictly to the future of $r_2$.\footnote{This behavior is precisely characterized by the Gao-Wald theorem. Under the null generic condition \eqref{eq:NGC}, the boundary null geodesic connecting $\beta_1$ and $r_2$ is prompt.}

This causal separation creates a \emph{causal shadow}: a region within the entanglement wedge $\ew(V)$ that lies outside $J^+[\hat{{D}}(V)] \cap J^-[\hat{{D}}(V)]$ and thus lacks causal contact with the boundary domain of dependence $\hat{{D}}[V]$.

Now consider a proper subset $V$ of a boundary Cauchy slice $\hat{\Sigma}$. For a pure global state, $V$ and its complement $V^c$ share the same HRRT surface\footnote{Recall that we use the same notation for a causal domain and its Cauchy surface.}. Let $\Sigma$ be a bulk Cauchy surface bounded by $\hat{\Sigma}$ that contains this common HRRT surface. The causal wedges $\cw(V)$ and $\cw(V^c)$ intersect $\Sigma$ in regions on opposite sides of this shared HRRT surface, indicating the presence of a causal shadow (which may be empty).

\subsection{Proof Strategy via Causal Anchoring Principle}\label{subsec:proof_strategy}
The Gao-Wald Theorem relates bulk and boundary causal structures for boundary points and trivially extends to a boundary compact set. One might expect similar relations for general compact sets $\mathcal{U} \subset M$ with bulk components:
\begin{equation}\label{eq:bulk_Gao_wald}
J^\pm[\mathcal{U}] \cap \partial M = \hat{J}^\pm[\mathcal{U} \cap \partial M].
\end{equation}
This fails for simple cases like bulk points and general compact sets with both bulk and boundary components. However, it holds for homology regions of boundary subsets $V$ (\ref{eq:RT_bb}). This strongly constrains bulk geometry in asymptotic scattering problems: null sheets emanating from HRRT surfaces and causal surfaces for boundary region $V$ are constrained by $V$'s boundary geometry.

Let us recall several useful lemmas and theorems from \cite{headrick2014causality}.
\begin{lemma}
Let $\bb$ denote the timelike boundary of an asymptotically global AdS spacetime $M$. Let $p\in \bb$ be an arbitrary point, $V$ be a spacelike region in a Cauchy slice $\hat{\Sigma}$ of $\bb$ and $\mathcal{V}$ be a proper subset of a Cauchy slice $\Sigma$ of $M$. Then,
\begin{align}
    J^{\pm}[p]\cap \bb &= \hat{J}^{\pm}[p],\label{eq:point_J} \\
    \ew(V)\cap \bb &= \hat{D}(V) \label{eq:ew_bb}, \\
    J^{\pm}[RT(V)]\cap \bb &= \hat{J}^{\pm}[\partial V] \label{eq:RT_bb},\\
    M &= D(\mathcal{V})\cup D(\mathcal{V}^c) \cup J^+[\partial \mathcal{V}] \cup J^-[\partial \mathcal{V}] \label{eq:D_J},
\end{align}
where $\mathcal{V}^c$ denotes the set complement of $\mathcal{V}$ on the Cauchy slice $\Sigma$ of $M$ and $\partial \mathcal{V}$ denotes the boundary of $\mathcal{V}$ in the bulk $M$.
\end{lemma}

Equation \eqref{eq:point_J} expresses the Gao-Wald Theorem: for a boundary causal domain $V = \hat{J}^-[p] \cap \hat{J}^+[q]$, the bulk causal wedge is $J^+[p] \cap J^-[q]$. Taking $c_1$ as an example, both the causal surface of $V_1$ and that of $W_1\cup Y_1\cup W_2$ lie on the null sheet $\partial J^+[c_1]$, equaling its intersection with appropriate bulk Cauchy slices.

Equations \eqref{eq:ew_bb} and \eqref{eq:RT_bb} generalize the Gao-Wald Theorem to homology regions. Specifically, null sheets emanating from HRRT surfaces (e.g., $RT(V_1)$) are anchored at $\hat{J}^\pm[\partial V]$ on $\partial M$. The same is true for null sheets emanating from causal surfaces, due to the causal wedge-entanglement wedge inclusion.

In an asymptotically global AdS spacetimes, matter/curvature distorts bulk null sheets $\mathcal{N}$ relative to their pure AdS counterparts $\mathcal{N}'$, but their boundary restrictions $\mathcal{N}\cap \partial M$ coincide. Our proof strategy therefore uses boundary null rays from relevant points to constrain the bulk geometry of entanglement wedges and causal wedges.

While \cite{LL2025superadditivity} proved their conjecture in pure $AdS_3$, matter/curvature in general asymptotically global AdS requires modifying the conjecture. Both the Gao-Wald Theorem and results from \cite{headrick2014causality} stem from causality, demonstrating that boundary causal consistency strongly constrains bulk geometry.

We end this subsection with a projection that will be used later.
\begin{remark}\label{rmk:project_global_hyperbolic}
We utilize the AdS-hyperbolicity of $M$ to identify points across different Cauchy slices of $M$ or its conformal compactification $\overline{M}$.

Recall that the AdS-hyperbolicity of $M$, or equivalently the global hyperbolicity of $\overline{M}$, implies that $\overline{M}$ has the topology $\Sigma \times \mathbb{R}$, where $\Sigma$ is a Cauchy surface of $\overline{M}$. As established in Theorems 8.3.14 and 8.2.2 of \cite{waldGR}, global hyperbolicity ensures the existence of a global time function $t$ (though highly non-unique). Each level set of $t$ is a Cauchy surface $\Sigma_t$, and the gradient $\nabla t$ defines a global timelike vector field. By projecting along the integral curves of $\nabla t$, we can map all spacetime points to a fixed Cauchy slice $\Sigma_{t_0}$.
\end{remark}

\subsection{Bulk Scattering region Contained in Entanglement Wedge}\label{subsec:S0_SE_inc}

We begin by providing a more careful proof of a lemma that was initially established in \cite{maypenington2020holographic} and further discussed in \cite{LL2025superadditivity}. The subtlety identified in Lemma \ref{lemma:A1}(2) proves advantageous for our proof here but precludes any conclusion that $\mathcal{S}_E = \tilde{\mathcal{S}}_E$.
\begin{figure}
    \centering
    \includegraphics[width=0.6\linewidth,trim={5 4cm 5 7},clip]{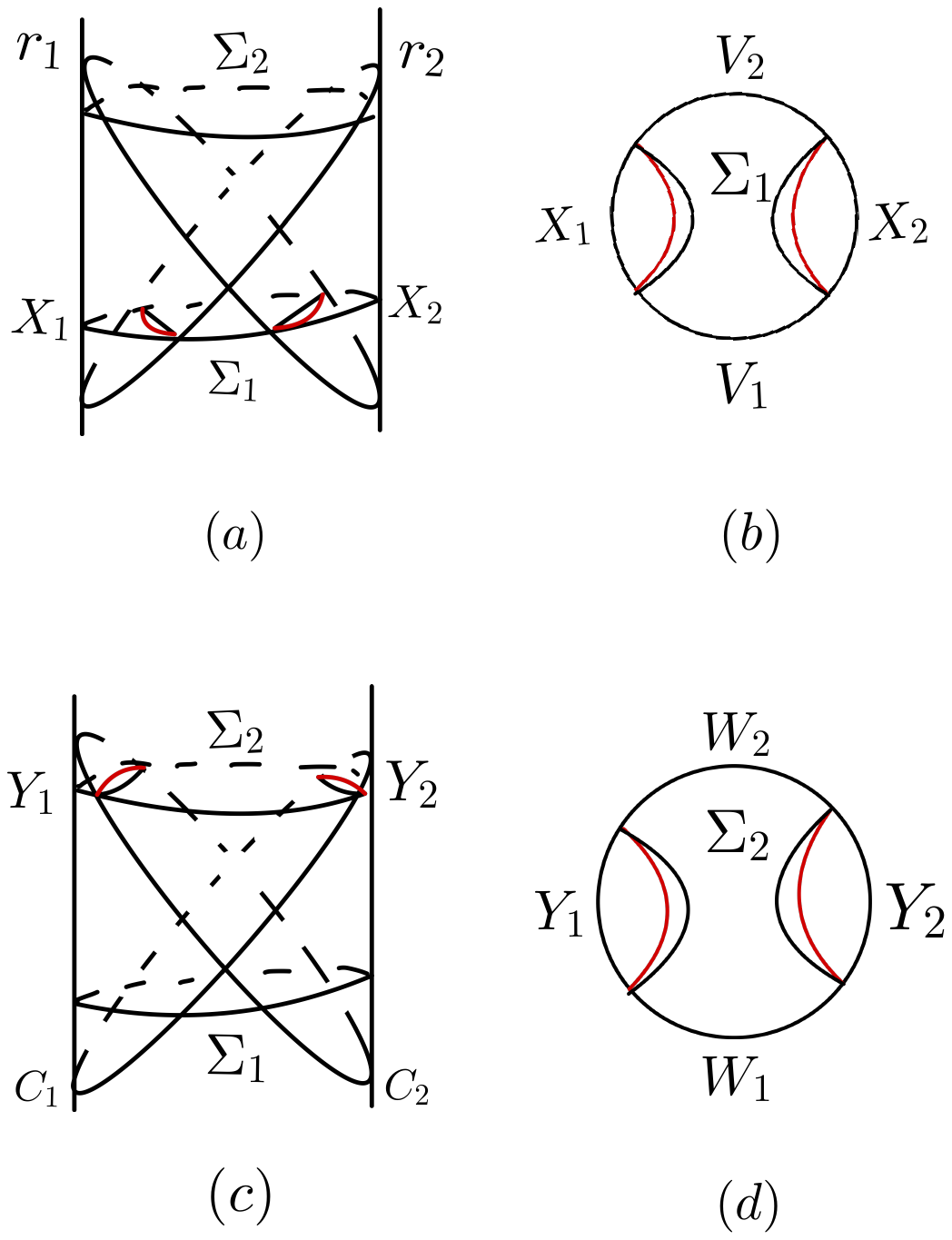}
    \caption{Schematic illustrating $\mathcal{S}_0 \subseteq \mathcal{S}_E$ when both $\ew(V_1 \cup V_2)$ and $\ew(W_1 \cup W_2)$ are connected. Panel $(a)$ depicts $\partial J^-[r_1]$ and $\partial J^-[r_2]$ and their intersections with the Cauchy slice $\Sigma_1$. Panel $(b)$ shows $\Sigma_1$ with the causal surface of $V_1 \cup X_i \cup V_2$ (i.e., $J^-[r_i] \cap \Sigma_1$) in black and the HRRT surface of $X_i$ in red. Panel $(c)$ displays $\partial J^+[c_1]$ and $\partial J^+[c_2]$ and their intersections with $\Sigma_2$. Panel $(d)$ presents $\Sigma_2$ with the causal surface of $W_1 \cup Y_i \cup W_2$ (i.e., $J^+[c_i] \cap \Sigma_2$) in black and the HRRT surface of $Y_i$ in red. Note that panels (c) and (d) are rotated relative to panels (a) and (b).}
    \label{fig:S0_SE}
\end{figure}

\begin{lemma}\label{lemma:S0_SE} 
Assume the conditions in Assumption \ref{assumption:1}.
    Define the bulk scattering region as
    $$\mathcal{S}_0 = J^+ [c_1] \cap J^+ [c_2] \cap J^- [r_1] \cap J^-[r_2] $$
    and the generalized scattering region\footnote{There is another bulk scattering region $S=J^+[\ew(V_1)]\cap J^+[\ew(V_2)]\cap J^-[\ew(W_1)] \cap J^-[\ew(W_2)]$ defined in \cite{may2021region}. Since the focus of this work is about $\mathcal{S}_E$, we do not discuss $S$ explicitly.} as 
    $$\mathcal{S}_E = \ew(V_1 \cup V_2) \cap  \ew(W_1 \cup W_2).$$ 
    Then $\mathcal{S}_0 \subseteq \mathcal{S}_E$.
\end{lemma}
\begin{proof}

If $\mathcal{S}_0 = \emptyset$, then the statement is trivial. Let us consider only $\mathcal{S}_0 \neq \emptyset$. Then by the original Connected Wedge Theorem, $\ew(V_1 \cup V_2)$ and $\ew(W_1 \cup W_2)$ are both connected.

    For the purpose of explanation, let us pick $\hat{\Sigma}_1$ to be a boundary Cauchy slice of $V_1 \cup V_2 \cup X_1 \cup X_2$ and $\Sigma_1$ be a corresponding bulk Cauchy slice containing relevant HRRT surfaces and causal surfaces\footnote{This is always possible because 1) a HRRT surface always lies in the causal shadow and hence, if different, is spacelike separated from relevant causal surfaces \cite{headrick2014causality}; 2) HRRT surfaces of disjoint, spacelike-separated boundary regions can be minimal on the same Cauchy slice \cite{wall2014maximin}.}. Similarly, let $\hat{\Sigma}_2$ be a boundary Cauchy slice of $W_1 \cup W_2 \cup Y_1 \cup Y_2$ and $\Sigma_2$ be a corresponding bulk Cauchy surface that contains relevant HRRT surfaces and causal surfaces.

    Let us first show that above $\Sigma_1$ the following is true:
    \begin{equation}\label{eq:S_SE_1}
        J^-[r_1]\cap J^-[r_2]\cap J^+[\Sigma_1] \subseteq \ew(V_1\cup V_2) \cap J^+[\Sigma_1].
    \end{equation} 
    On the boundary $\bb$ and above $\hat{\Sigma}_1$, 
    \begin{equation}
        \hat{J}^-[r_1] \cap \hat{J}^+[\hat{\Sigma}_1] = \hat{D}(V_1 \cup X_1\cup V_2) \cap \hat{J}^+[\hat{\Sigma}_1].
    \end{equation}
    As noted in Section \ref{subsec:proof_strategy}, by definition of causal wedge and \eqref{eq:point_J}, it follows that in the bulk
    \begin{equation}
        J^-[r_1] \cap J^+[\Sigma_1] = \cw(V_1 \cup X_1 \cup V_2) \cap J^+[\Sigma_1]
    \end{equation}
    and 
    \begin{equation}\label{eq:S0_SE_a}
        J^-[r_1] \cap \Sigma_1 = \cw(V_1\cup X_1 \cup V_2)\cap \Sigma_1 \subseteq \ew (V_1 \cup X_1 \cup V_2)\cap \Sigma_1.
    \end{equation}
    Similarly, by considering $r_2$ and $X_2$, we  have
    \begin{equation}\label{eq:S0_SE_b}
        J^-[r_2] \cap \Sigma_1 = \cw(V_1\cup X_2 \cup V_2)\cap \Sigma_1 \subseteq \ew (V_1 \cup X_2 \cup V_2)\cap \Sigma_1.
    \end{equation}
    Since $\ew(V_1\cup V_2)$ is connected, a portion of $RT(V_1\cup V_2)$ is shared with $RT(V_1\cup X_1\cup V_2)$ while another portion of $RT(V_1\cup V_2)$ is shared with $RT(V_1\cup X_2\cup V_2)$. In other words,
    \begin{equation}
        \ew (V_1 \cup X_2 \cup V_2) \cap \ew (V_1 \cup X_2 \cup V_2)\cap \Sigma_1 = \ew (V_1 \cup V_2)\cap \Sigma_1.
    \end{equation}
    Therefore, taking the intersection of \eqref{eq:S0_SE_a} and \eqref{eq:S0_SE_b}, we get
    \begin{equation}\label{eq:S0_SE_c}
        J^-[r_1]\cap J^-[r_2] \cap \Sigma_1 \subseteq \ew(V_1\cup V_2) \cap \Sigma_1.
    \end{equation}
    This is schematically depicted in Fig.\ \ref{fig:S0_SE}(a,b).
    
    As shown in Lemma \ref{lemma:A1}, future-pointing and inward-pointing (i.e. pointing toward the bulk) null sheets emanating from $\partial J^-[r_1]\cap \Sigma_1$ and $\partial J^-[r_2]\cap \Sigma_1$ (black curves in Figure\ \ref{fig:S0_SE}(a,b)), if different from, lie to the \textit{future} of $\partial J^-[r_1]$ and $\partial J^-[r_2]$.
    As noted in Section \ref{subsec:proof_strategy}, the future horizon of $\ew(V_1\cup V_2)$ consists of future-pointing and inward-pointing null sheets emanating from $RT(V_1\cup V_2)$, which in turn lie to the \textit{future} of future-pointing and inward-pointing null sheets emanating from $\partial J^-[r_1] \cap \Sigma_1$ and $\partial J^-[r_2] \cap \Sigma_1$ because of \eqref{eq:S0_SE_c}.
    Combining these facts, we get \eqref{eq:S_SE_1}.

    The reasoning for 
    \begin{equation}\label{eq:S_SE_2}
        J^+[c_1] \cap J^+[c_2] \cap J^-[\Sigma_2] \subseteq \ew(W_1 \cup W_2) \cap J^-[\Sigma_2] 
    \end{equation} 
    is exactly the same.
    On the boundary $\bb$ and below $\hat{\Sigma}_2$, 
    \begin{equation}
        \hat{J}^+[c_1] \cap \hat{J}^-[\hat{\Sigma}_2] = \hat{D}(W_1 \cup Y_1\cup W_2) \cap \hat{J}^-[\hat{\Sigma}_2].
    \end{equation}
    As noted in Section \ref{subsec:proof_strategy}, by definition of causal wedge and \eqref{eq:point_J}, it follows that in the bulk
    \begin{equation}\label{eq:S0_SE_d}
        J^+[c_1] \cap \Sigma_2 = \cw(W_1\cup Y_1 \cup W_2)\cap \Sigma_2 \subseteq \ew (W_1 \cup Y_1 \cup W_2)\cap \Sigma_2.
    \end{equation}
    Similarly, by considering $c_2$ and $Y_2$, we  have
    \begin{equation}\label{eq:S0_SE_e}
        J^+[c_2] \cap \Sigma_2 = \cw(W_1\cup Y_2 \cup W_2)\cap \Sigma_2 \subseteq \ew (W_1 \cup Y_2 \cup W_2)\cap \Sigma_2.
    \end{equation}
    Taking the intersection of \eqref{eq:S0_SE_d} and \eqref{eq:S0_SE_e}, we get
        $$J^+[c_1] \cap J^+[c_2]\cap \Sigma_2 \subseteq \ew(W_1\cup W_2) \cap \Sigma_1,$$
    where we have used that fact that $\ew(W_1\cup W_2)$ is connected.
    This is schematically depicted in Fig.\ \ref{fig:S0_SE}(c,d). The rest follows directly.

    Combining \eqref{eq:S_SE_1} and \eqref{eq:S_SE_2}, we get that
    $$ J^+ [c_1] \cap J^+ [c_2] \cap J^- [r_1] \cap J^-[r_2] \cap J^+[\Sigma_1] \cap J^-[\Sigma_2] \subseteq \ew(V_1 \cup V_2) \cap  \ew(W_1 \cup W_2) \cap J^+[\Sigma_1] \cap J^-[\Sigma_2] $$
    Note that
    $\ew(V_1 \cup V_2) \cap  \ew(W_1 \cup W_2) \cap J^+[\Sigma_1] \cap J^-[\Sigma_2] \subseteq \ew(V_1 \cup V_2) \cap  \ew(W_1 \cup W_2)$ always holds, so we can drop out $ J^+[\Sigma_1] \cap J^-[\Sigma_2]$ on the right hand side.
    We show below that $\mathcal{S}_0=J^+ [c_1] \cap J^+ [c_2] \cap J^- [r_1] \cap J^-[r_2]$ always lies to the future of $\Sigma_1$ and to the past of $\Sigma_2$, so we can also drop out $ J^+[\Sigma_1] \cap J^-[\Sigma_2]$ on the left hand side. Therefore, we get the claim
    \begin{equation}
        J^+ [c_1] \cap J^+[c_2] \cap J^-[r_1] \cap J^-[r_2] \subseteq \ew(V_1 \cup V_2) \cap  \ew(W_1 \cup W_2).
    \end{equation}
    
    Assume by contradiction that $J^+[c_1] \cap J^+[c_2]\cap J^-[\Sigma_1] \neq \emptyset$, then we must have $\cw(V_1)\cap \cw(V_2)\neq \emptyset$. Since the entanglement wedge contains the causal wedge, $\ew(V_1)\cap \ew(V_2)\neq \emptyset$.
    But this contradicts the entanglement wedge nesting property, which implies that disjoint spacelike regions $V_1$ and $V_2$ have spacelike separated HRRT surfaces \cite{wall2014maximin}.
    Similarly, $J^-[r_1]\cap J^-[r_2]\cap J^+[\Sigma_2]=\emptyset$. Therefore, 
    \begin{align*}
        & J^+[c_2] \cap J^+[c_2] \cap J^-[r_1] \cap J^-[r_2] \cap J^+[\Sigma_1] \cap J^-[\Sigma_2] \\
        &= J^+[c_2] \cap J^+[c_2] \cap J^-[r_1] \cap J^-[r_2] \cap [J^-[\Sigma_1]\cup J^+[\Sigma_1]] \cap [J^-[\Sigma_2] \cup J^+[\Sigma_2]] \\
        &= J^+[c_2] \cap J^+[c_2] \cap J^-[r_1] \cap J^-[r_2].
    \end{align*}
\end{proof}

\subsection{Generalized Connected Wedge Theorem}\label{subsec:GCWT}
We now proceed to the generalized Connected Wedge Conjecture proposed by \cite{LL2025superadditivity}, which we state again for the convenience of the reader.
\begin{conjecture}\label{conj:LL}
Assume the conditions in Assumption \ref{assumption:1},
    $\mathcal{S}_E\neq \emptyset$ if and only if $\ew(V_1\cup V_2)$ is connected.
\end{conjecture}

\begin{figure}
    \centering
    \includegraphics[width=0.8\linewidth,trim={5 15cm 5 3cm},clip]{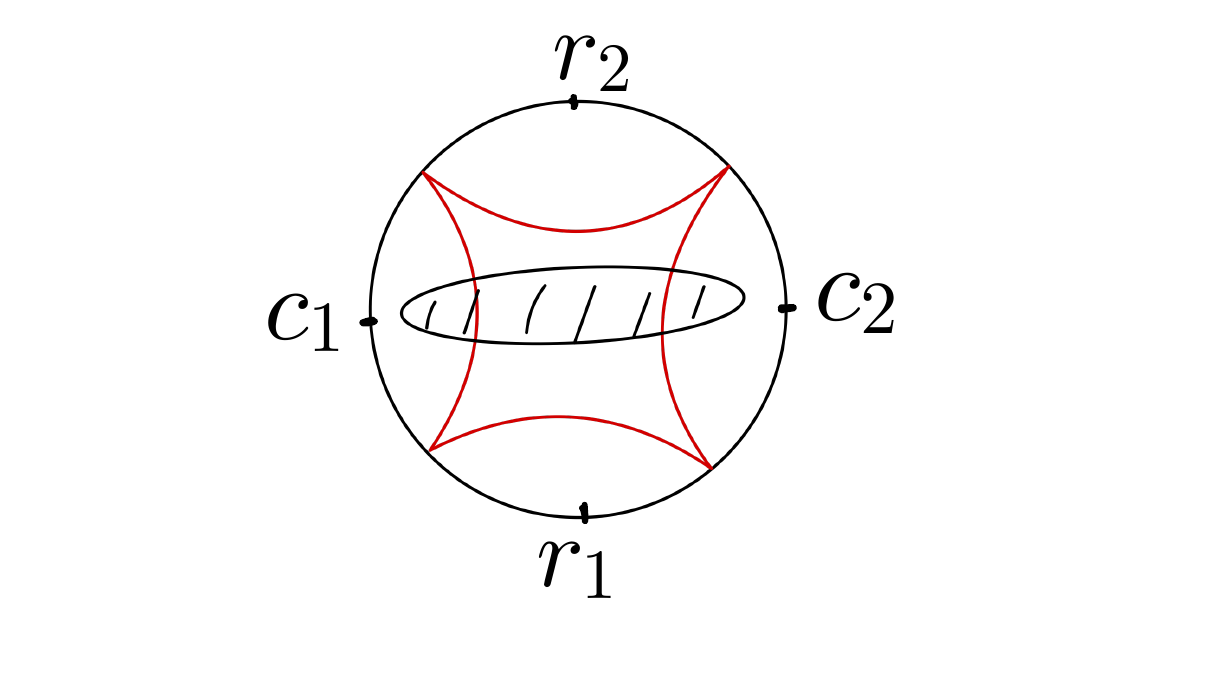}
    \caption{An explicit example of $\ew(V_1\cup V_2)$ being connected while $\ew(W_1\cup W_2)$ being disconnected. Here, the input points $c_1,c_2$ and output points $r_1,r_2$ as well as RT surfaces of $V_1, V_2$ and $W_1, W_2$ are all projected onto a single Cauchy slice (see Remark \ref{rmk:project_global_hyperbolic}). The four points $c_1,c_2,r_1,r_2$ are equally spaced along the circle. The shaded region indicates addition of matter in a non-spherically-symmetric way while maintaining the reflection symmetry across segment $r_1-r_2$ and segment $c_1-c_2$. This will make HRRT surfaces of $W_1, W_2$ (on $\Sigma_2$) and HRRT surfaces of $X_1, X_2$ (on $\Sigma_1$) to be the true minimum.}
    \label{fig:Vd_Wc_ex}
\end{figure}

We remark that Conjecture \ref{conj:LL} would have a corollary that $\ew(V_1\cup V_2)$ is connected if and only if $\ew(W_1\cup W_2)$ is connected, due to the symmetry between $V_i$'s and $W_i$'s in the definition of $\mathcal{S}_E$. However, we give an explicit example of $\ew(V_1\cup V_2)$ being connected yet $\ew(W_1\cup W_2)$ being disconnected. Examples of $\ew(V_1\cup V_2)$ being disconnected yet $\ew(W_1\cup W_2)$ being connected can be constructed similarly.
\begin{example}
   Start with the vacuum $AdS_3$ and consider the marginal set-up of $c_1,c_2, r_1, r_2$ such that $V_1, V_2, X_1, X_2$ and $W_1, W_2, Y_1, Y_2$ are all of the same size. Now consider a modified bulk metric that is time-translation invariant yet increases the measure along the $c_1-c_2$ direction. For example, one can add to the vacuum $AdS_3$ matter of ellipsoid shape which is elongated along the $c_1-c_2$ direction (Figure \ref{fig:Vd_Wc_ex}). This will make $\ew(V_1\cup V_2)$ and $\ew(Y_1\cup Y_2)$ both connected. Since $RT(W_1\cup W_2)=RT(Y_1\cup Y_2)$, we reach the situation where $\ew(V_1\cup V_2)$ is connected while $\ew(W_1\cup W_2)$ is disconnected. 

   We expect that there are generic situations in which one of $\ew(V_1 \cup V_2)$ or $\ew(W_1 \cup W_2)$ is connected while the other is disconnected, and that this phenomenon is not inherently tied to the symmetric configuration employed in the above example. 
\end{example}

This provides a counterexample to a corollary of Conjecture \ref{conj:LL}. Therefore, we propose a slight modification to the conjecture. 
First, we prove the most nontrivial part of Conjecture \ref{conj:LL}: if both $\ew(V_1\cup V_2)$ and $\ew(W_1\cup W_2)$ are connected, then they have a non-empty intersection.

Rephrasing in the language of operator algebra or quantum information, if the two input regions share $O(1/G_N)$ mutual information and if there exists a time-reversal symmetry so that the two output regions also share $O(1/G_N)$ mutual information, then there exists an operator common to the large-$N$ limit operator algebras $\mm_{V_1\cup V_2}$ and $\mm_{W_1\cup W_2}$.

\begin{theorem}\label{thm:SE_wedge}
Assume conditions in Assumption \ref{assumption:1}.
    Then, $\mathcal{S}_E =\ew(V_1\cup V_2)  \cap \ew(W_1\cup W_2)$  has positive measure if both $\ew(V_1 \cup V_2)$ and $\ew(W_1 \cup W_2)$ are connected.
\end{theorem}

\begin{proof}
For the purpose of explanation, let us pick $\hat{\Sigma}_1$ to be a boundary Cauchy slice of $V_1 \cup V_2 \cup X_1 \cup X_2$ and $\Sigma_1$ to be a corresponding bulk Cauchy slice that contains relevant HRRT surfaces and causal surfaces\footnote{This is always possible because 1) a HRRT surface always lies in the causal shadow and hence, if different, is spacelike separated from relevant causal surfaces \cite{headrick2014causality}; 2) HRRT surfaces of disjoint, spacelike-separated boundary regions can be minimal on the same Cauchy slice \cite{wall2014maximin}.}. Similarly, let $\hat{\Sigma}_2$ be a boundary Cauchy slice of $W_1 \cup W_2 \cup Y_1 \cup Y_2$ and $\Sigma_2$ be a corresponding bulk Cauchy slice that contains relevant HRRT surfaces and causal surfaces.

\begin{figure}
    \centering
    \includegraphics[width=0.8\linewidth]{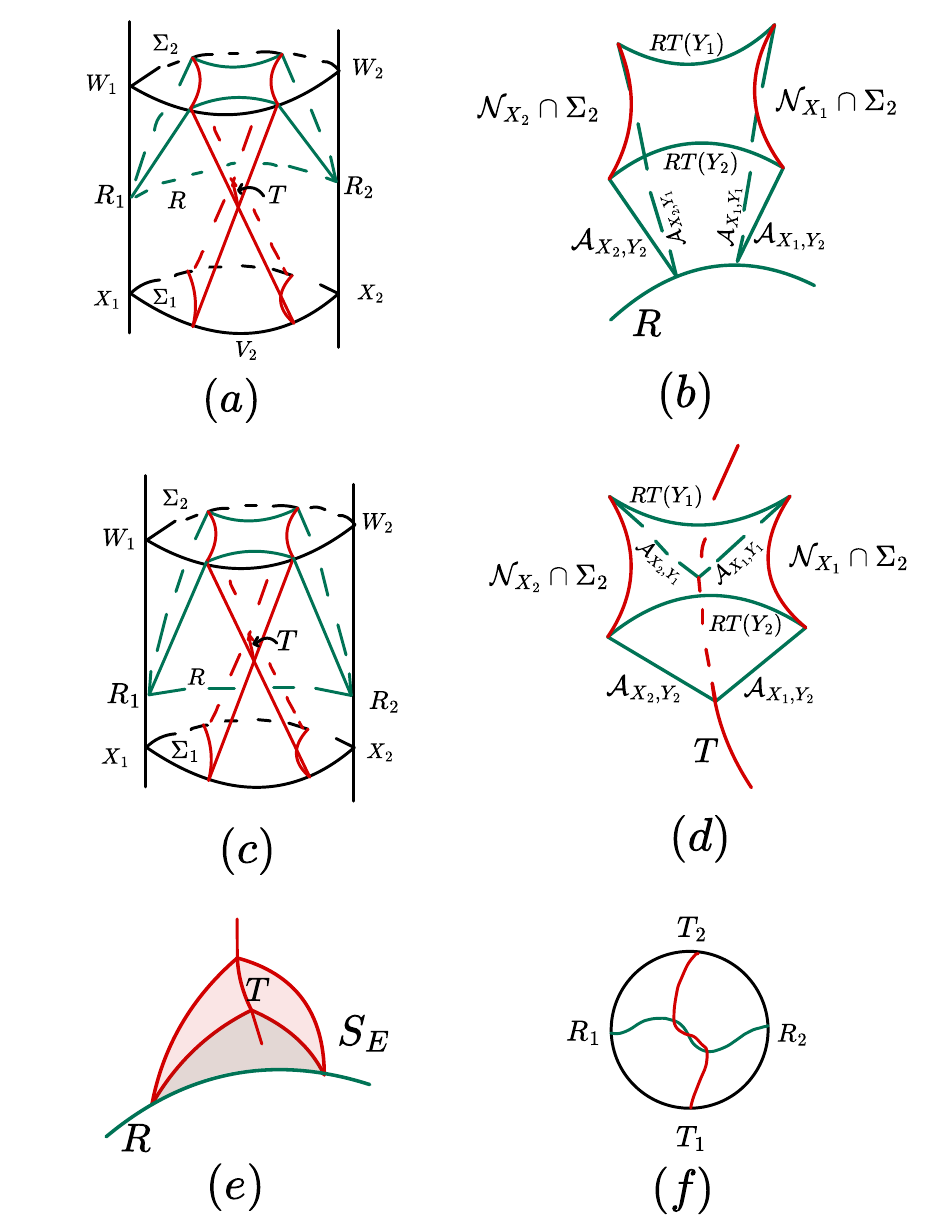}
    \caption{Schematic for intersections among $\NN_{X_i}$ and $\NN_{Y_j}$ when $\ew(V_1\cup V_2)$ and $\ew(W_1 \cup W_2)$ are both connected. The past-pointing null sheets from $RT(Y_1)$ and $RT(Y_2)$ intersect at the ridge $R$, which intersects $\bb$ at points $R_1$ and $R_2$. The future-pointing null sheets from $RT(X_1)$ and $RT(X_2)$ intersect at the ridge $T$, which intersects $\bb$ at points $T_1$ and $T_2$. A case of $\mathcal{S}_E = \emptyset$ is shown in (a-b) while a case of $\mathcal{S}_E$ has positive measure is shown in (c-e) with (e) illustrating the geometry of $\mathcal{S}_E$. Shown in $(f)$ is an example that ridges $R$ and $T$ or their projections intersect more than once. Red curves are associated with $\NN_{X_i}$ while green curves are associated with $\NN_{Y_j}$.}
    \label{fig:cut_ridge}
\end{figure}

Since $ \ew(V_1 \cup V_2)$ and $\ew(W_1 \cup W_2) $  are both connected, $RT(V_1\cup V_2)=RT(X_1)\cup RT(X_2)$ and $RT(W_1\cup W_2)=RT(Y_1)\cup RT(Y_2)$.
    We consider a similar geometric structure as used in ref. \cite{maypenington2020holographic}. Consider the two past-pointing and inward-pointing (e.g. pointing toward the bulk) null sheets emanating from $RT(Y_1)$ and $RT(Y_2)$ (green in Figure \ref{fig:cut_ridge}(a,c)) and the two future-pointing and inward-pointing (e.g. pointing toward the bulk) null sheets emanating from $RT(X_1)$ and $RT(X_2)$ (red in Figure \ref{fig:cut_ridge}(a,c)). Since we assume that the global boundary state is a pure state, $RT(Y_1)=RT(W_1\cup Y_2 \cup W_2)$. It follows from \eqref{eq:RT_bb} that the past-pointing and inward-pointing null sheet $\mathcal{N}_{Y_1}$ emanating from $RT(Y_1)$ is the past horizon of $\ew(W_1\cup Y_2\cup W_2)$ \footnote{Recall the past horizon of an entanglement wedge $\ew$ is the past Cauchy horizon of $\ew$ as a domain of dependence.}. Moreover, as emphsized in Section \ref{subsec:proof_strategy}, $\mathcal{N}_{Y_1}\cap \bb$ is the past horizon of $\hat{D}(W_1\cup Y_2\cup W_2)$, i.e. two boundary null rays emanating from $\partial Y_1$ (pointing away from $Y_1$ on $\bb$, red curves in Figure \ref{fig:setup}). Recall the definition of $Y_1$ and $\alpha_1$ in section \ref{subsec:review}, we have
    \begin{equation}\label{eq:NY1_bb}
        \mathcal{N}_{Y_1}\cap \bb \subseteq \hat{J}^-[\alpha_1].
    \end{equation}
    For the past-pointing and inward-pointing null sheet $\mathcal{N}_{Y_2}$ emanating from $RT(Y_2)$, we have
    \begin{equation}\label{eq:NY2_bb}
        \mathcal{N}_{Y_2}\cap \bb \subseteq \hat{J}^-[\alpha_2].
    \end{equation}    
  Similarly, the future pointing and inward-pointing null sheet $\mathcal{N}_{X_1}$ emanating from $RT(X_1)$ is the future horizon of $\ew(V_1\cup X_2\cup V_2)$ and intersects $\bb$ at boundary null rays from $\partial X_1$ (pointing away from $X_1$). Recall the definition of $X_1$ and $\beta_1$ in section \ref{subsec:review}, we have
    \begin{equation}\label{eq:NX1_bb}
        \mathcal{N}_{X_1}\cap \bb \subseteq \hat{J}^+[\beta_1].
    \end{equation}
  Similarly for the future-pointing and inward-pointing null sheet $\mathcal{N}_{X_2}$ emanating from $RT(X_2)$, we have 
    \begin{equation}\label{eq:NX2_bb}
        \mathcal{N}_{X_2}\cap \bb \subseteq \hat{J}^+[\beta_2]. 
    \end{equation}  

    Note that causality further constrains the intersection of these four null sheets\footnote{In general, the intersection of two distinct null sheets or two arbitrary hypersurfaces could consist of multiple curves and isolated points (where they are tangent to each other)}.
    Specifically, the null sheet $\mathcal{N}_{Y_1}$ simply partitions the part of spacetime $J^+[\Sigma_1] \cap J^+[\Sigma_2]\subseteq M$ into two parts: one to its future and the other to its past. Similar is true for the other three null sheets. Then the two null sheets $\mathcal{N}_{Y_1}$ and $\mathcal{N}_{Y_2}$ simply partition the bulk spacetime region $J^+[\Sigma_1] \cap J^+[\Sigma_2]$ into four parts and intersect at a simple curve \footnote{If the two null sheets intersect at more than one curve or their intersecting curve has nonempty self-intersections, this would contradict that they partition $J^+[\Sigma_1] \cap J^+[\Sigma_2]$ into four parts.}. Note that $\partial \hat{J}^-[\alpha_1] \cap \partial \hat{J}^-[\alpha_2]$ consists of two points, $R_1$ and $ R_2$, in $\hat{J}^+[\hat{\Sigma}_1] \cap \hat{J}^-[\hat{\Sigma}_2] \subseteq \bb$ (Figure \ref{fig:setup}). We thus have $\mathcal{N}_{Y_1} \cap \mathcal{N}_{Y_2} \cap \bb $ consist of exactly these two points. Similarly $\mathcal{N}_{X_1} \cap \mathcal{N}_{X_2}$ must be a simple (non-self-intersecting) curve intersecting $\bb$ at the two intersection points of $\partial \hat{J}^+[\beta_1]$ and $\partial \hat{J}^-[\beta_2]$.
        
    Let $R$ denote the ridge of intersection between the two past-pointing null sheets emanating from $RT(Y_1)$ and $RT(Y_2)$, i.e. $R = \mathcal{N}_{Y_1}\cap \mathcal{N}_{Y_2}$. From above discussions, we have
    \begin{equation}
        R\cap \bb = \hat{J}^-[\alpha_1] \cap \hat{J}^-[\alpha_2]=\{R_1, R_2\}.
    \end{equation}
    Let $T$ denote the ridge of intersection between the two future-pointing null sheets emanating from $RT(X_1)$ and $RT(X_2)$, i.e. $T = \mathcal{N}_{X_1}\cap \mathcal{N}_{X_2}$. From above discussions, we have
    \begin{equation}
        T \cap \bb = \hat{J}^+[\beta_1] \cap \hat{J}^+[\beta_2] = \{T_1, T_2\}.
    \end{equation}
    As can be seen in Figure \ref{fig:setup}, $R_1, T_2, R_2, T_1$ are future boundary points of causal domains $X_1, V_2, X_2, V_1$, respectively. When we project the four points onto a single Cauchy slice, using a global time function of $\overline{M}$ (see Remark \ref{rmk:project_global_hyperbolic}), they would be ordered in cyclic order (clockwise or counter-clockwise) since $X_1, V_2, X_2, V_1$ are arranged so. This is illustrated in Figure \ref{fig:cut_ridge}(f). These considerations force the intersection number between the ridge $R$ and the ridge $T$, when projected onto the same Cauchy slice (Remark \ref{rmk:project_global_hyperbolic}), to be an odd integer (a standard explanation is included in Lemma \ref{lemma:A2} for completeness).
    
Then, the possible configurations of these four null sheets are as follows:    \begin{enumerate}
        \item[(1)] The ridge $R$ lies above the ridge $T$ (Figure \ref{fig:cut_ridge}(a-b)). More precisely, looking toward the future from $\Sigma_1$, the two future-pointing null sheets $\NN_{X_1},\NN_{X_2}$ from $RT(X_1)$ and $RT(X_2)$ intersect at the ridge $T$ before cutting a segment of ridge $R$. In this case, $\mathcal{S}_E=\emptyset$ since $\ew(V_1\cup V_2)$ lies to the past of $T$ while $\ew(W_1\cup W_2)$ lies to the future of $R$ \footnote{Recall that an entanglement wedge is bounded by causal boundaries of its RT surface, as noted in Section \ref{subsec:proof_strategy}.}. 
        \item[(2)] The ridge $R$ lies below the ridge $T$ (Figure \ref{fig:cut_ridge}(c-d)). More precisely, looking toward the past from $\Sigma_2$, the two past-pointing null sheets $\NN_{Y_1},\NN_{Y_2}$ from $RT(Y_1)$ and $RT(Y_2)$ cut a segment along the ridge $T$ before intersecting at the ridge $R$. In this case, $\mathcal{S}_E$ has positive measure since $\ew(V_1\cup V_2)$ lies to the past of $T$ while $\ew(W_1\cup W_2)$ lies to the future of $R$.
        \item[(3)] The ridge $R$ and the ridge $T$ intertwines. Their projection onto $\Sigma_1$ intersect at multiple points $p_1,...p_k, k\geq 3$: at some $p_i$, $R$ is above $T$ while at other $p_j$, $R$ is below $T$. In this case, $\mathcal{S}_E$ still has positive measure.
        \item[(4)] The marginal case that $R$ and $T$ themselves, not their projections, intersect at isolated points. In this case, $|\mathcal{S}_E|=0$. A particular subcase is when the ridge $R$ and the ridge $T$ and hence all four null sheets intersect at a single point.
    \end{enumerate}
    We employ a similar calculation as in \cite{maypenington2020holographic} to exclude Possibility $(1)$ and $(4)$, i.e. the possibility of $\mathcal{S}_E=\emptyset$ or $|\mathcal{S}_E|=0$, under the assumption that both $ \ew(V_1 \cup V_2)$ and $\ew(W_1 \cup W_2) $ are connected. 
    
    Recall that both future-pointing and past-pointing null geodesics emanating from extremal surfaces have non-positive expansion $\theta$, when choosing the affine parameter $\lambda$ to increase away from extremal surfaces. Also recall that on a null hypersurface $\mathcal{H}$, changes in the area of spacelike cross-sections along null generators satisfy
    \begin{equation}\label{eq:theta_area}
        \text{area}(\mathcal{S}_2)-\text{area}(\mathcal{S}_1) = \int_\mathcal{H} \theta
    \end{equation}
    where the affine parameter of the cross-section $\mathcal{S}_2$ is greater than that of the cross-section $\mathcal{S}_1$.

    Consider Possibility $(1)$. Denote the part of the ridge $R$ cut by future pointing null sheets from $RT(X_1)$ and $RT(X_2)$ by $R_X$ (Figure \ref{fig:cut_ridge}(b)). By abuse of notation, we denote the subset of $\NN_{X_1}$ appearing in Figure \ref{fig:cut_ridge}(b) still by $\NN_{X_1}$ and similarly for the other three. Denote the cusps where a pair of the four null sheets intersect by $\mathcal{A}$ with proper subscripts. For example, the cusp formed by $\NN_{X_1}\cap \NN_{Y_2}$ is denoted by $\mathcal{A}_{X_1,Y_2}$. Denote the set of focal points and crossover seams on relevant null hypersurfaces by $F$ with suitable subscripts, e.g. $F_{X_1}$ on $\mathcal{N}_{X_1}$. These focal points and crossover seams are exactly where null geodesics leave the causal boundary and hence are also the set of null geodesic endpoints on these null hypersurfaces.
    
    Consider the two past-pointing null sheets emanating from $RT(Y_1)$ and $RT(Y_2)$ until they intersect at the ridge $R$, we have
    \begin{equation}\label{eq:ridge_1}
|\mathcal{A}_{X_1,Y_1}\cup \mathcal{A}_{X_1,Y_2}\cup \mathcal{A}_{X_2,Y_1}\cup \mathcal{A}_{X_2,Y_1}| +  |F_{Y_1}|+|F_{Y_2}|+|R_X| - |RT(Y_1 \cup Y_2)| = \int_{\mathcal{N}_{Y_1}\cup \mathcal{N}_{Y_2}} \theta \leq 0
    \end{equation}
    where $|\cdot|$ denotes the length.
    Consider the two future-pointing null sheets $\NN_{X_1},\NN_{X_2}$ emanating from $RT(X_1)$ and $RT(X_2)$ starting from the ridge $R$ until they intersect $\Sigma_2$, we have
    \begin{equation}\label{eq:ridge_2}
       |\NN_{X_1}\cap \Sigma_2| +  |\NN_{X_2} \cap \Sigma_2| + |F_{X_1}| +|F_{X_2}| - |\mathcal{A}_{X_1,Y_1}\cup \mathcal{A}_{X_1,Y_2}\cup \mathcal{A}_{X_2,Y_1}\cup \mathcal{A}_{X_2,Y_1}|  = \int_{\mathcal{N}_{X_1}\cup \mathcal{N}_{X_2}} \theta \leq 0.
    \end{equation}
    Combining \eqref{eq:ridge_1} and \eqref{eq:ridge_2}, we get
    \begin{equation}\label{eq:ridge_3}
        |\NN_{X_1}\cap \Sigma_2| +  |\NN_{X_2} \cap \Sigma_2| < |RT(Y_1 \cup Y_2)| 
    \end{equation}
    where the inequality is strict particularly because $|R_X|>0$.
    Note that \eqref{eq:NX1_bb} implies that $\NN_{X_1}\cap \Sigma_2$ have the same endpoints on $\bb$ as $RT(W_2)$, i.e. $\NN_{X_1}\cap \Sigma_2$ is homologous to $RT(W_2)$. Similarly, \eqref{eq:NX2_bb} implies that $\NN_{X_2}\cap \Sigma_2$ is homologous to $RT(W_1)$. Then, \eqref{eq:ridge_3} contradicts the assumption that $\ew(W_1\cup W_2)$ is connected, which instead implies that
    \begin{equation}
       |RT(Y_1)| + |RT(Y_2)| <  |RT(W_1)| + |RT(W_2)| \leq  |\NN_{X_1}\cap \Sigma_2| +  |\NN_{X_2} \cap \Sigma_2|.
    \end{equation}
    This excludes Possibility $(1)$ that $\mathcal{S}_E=\emptyset$.

    For Possibility (2) (Figure \ref{fig:cut_ridge}(d)), the same analysis would yield
    \begin{align}\label{eq:ridge_case2}
        &|\mathcal{A}_{X_1,Y_1}\cup \mathcal{A}_{X_1,Y_2}\cup \mathcal{A}_{X_2,Y_1}\cup \mathcal{A}_{X_2,Y_1}|  + |F_{Y_1}| +|F_{Y_2}| - |RT(Y_1 \cup Y_2)| \nonumber \\
        &= \int_{\mathcal{N}_{Y_1}\cup \mathcal{N}_{Y_2}}\theta \leq 0, \\
        &|\NN_{X_1}\cap \Sigma_2| +  |\NN_{X_2} \cap \Sigma_2| + |F_{X_1}|+|F_{X_2}| - |\mathcal{A}_{X_1,Y_1}\cup \mathcal{A}_{X_1,Y_2}\cup \mathcal{A}_{X_2,Y_1}\cup  \mathcal{A}_{X_2,Y_1}| - |T_Y| \nonumber \\
        &= \int_{\mathcal{N}_{X_1}\cup \mathcal{N}_{X_2}} \theta \leq 0,
    \end{align}
    where $T_Y$ denote the part of $T$ cut by the two past pointing null sheets from $RT(Y_1)$ and $RT(Y_2)$. There is no contradiction in this case, particularly because of the nonzero $|T_Y|$. 
    
    For Possibility (3), we have a combination of the previous two cases. A similar analysis does not lead to contradiction. We leave out the details to the reader. 
    
    For possibility (4), if the ridge $R$ and the ridge $T$ have multiple intersection points, then we obtain a similar geometric structure as that discussed in Possibility $(1)$ from the part of $R$ between the first intersection point and the last intersection point (Figure \ref{fig:cut_ridge_marginal_multip}(b)). We therefore exclude this scenario with a similar contradiction.       
\begin{figure}
    \centering
    \includegraphics[width=0.5\linewidth,trim={5 5cm 5 4cm},clip]{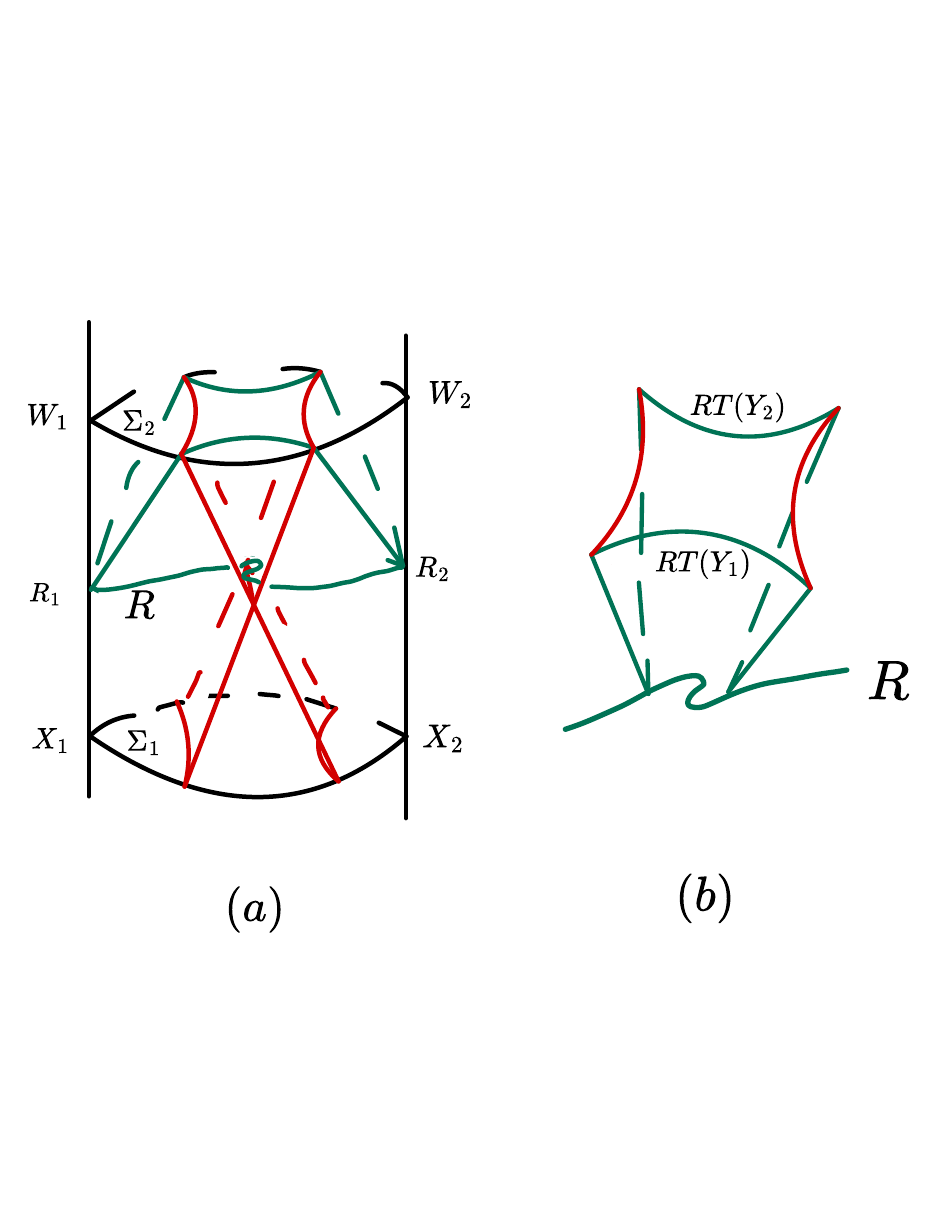}
    \caption{Schematic for $\mathcal{S}_E$ consisting of multiple points. That is, the ridge $R$ and the ridge $T$ intersect at more than one point. We get the same geometric structure as before when restricting to the prat of $R$ between the first and the last intersection point.}
    \label{fig:cut_ridge_marginal_multip}
\end{figure}

When the ridge $R$ and the ridge $T$ intersect at exactly one point $p$, $\mathcal{S}_E=\{p\}$ and $|R_X|=|T_Y|=0$.  Combining \eqref{eq:ridge_1} and \eqref{eq:ridge_2}, we have
\begin{align}
    |RT(W_1)|+|RT(W_2)| + |F| \leq |\NN_{X_1}\cap \Sigma_2| +  |\NN_{X_2} \cap \Sigma_2| + |F| \nonumber \\
    = \int_{\mathcal{N}_{X_1}\cup \mathcal{N}_{X_2}\cup \mathcal{N}_{Y_1}\cup \mathcal{N}_{Y_2}} \theta + |RT(Y_1\cup Y_2)| \leq |RT(Y_1)| + |RT(Y_2)|, \label{eq:SE_1p}
\end{align}
where $|F|$ include all null geodesic endpoints on all four partial null sheets. 
One can reach the same conclusion using \eqref{eq:ridge_case2}.
Equation \eqref{eq:SE_1p} implies that $\ew(W_1\cup W_2)$ is either disconnected or marginal, contradicting the assumption that $\ew(W_1 \cup W_2)$ is connected. Possibility $(4)$ is thus excluded. This completes the proof.
\end{proof}

Now to settle the Conjecture \ref{conj:LL}, we need to examine whether $\mathcal{S}_E = \emptyset$ follows when either of $\ew(V_1\cup V_2)$ and $\ew(W_1\cup W_2)$ is not connected. Ref. \cite{lima2025sufficientGCWT} found explicit counterexamples to such claims. Therefore, $\mathcal{S}_E\neq \emptyset$ is only a sufficient but not necessary condition to connectedness of the two entanglement wedges. Recall that the original CWT gives a necessary but sufficient condition, $\mathcal{S}_0 \neq \emptyset$, for both $\ew(V_1\cup V_2)$ and $\ew(W_1 \cup W_2)$ being connected. Therefore, the equivalence condition to connectedness of the two entanglement wedges must involve another set $\mathcal{S}'$ satisfying $\mathcal{S}_0 \subseteq \mathcal{S}' \subseteq \mathcal{S}_E$. Ref.\cite{lima2025sufficientGCWT} proposed $\mSE$. That  $\mSE\neq \emptyset$ is a necessary condition is immediate from its definition. We now extend the above proof of Theorem \ref{thm:SE_wedge} to show that $\mSE\neq \emptyset$ is also sufficient. We then illustrate on the relation between $\mSE$ and $\mathcal{S}_E$, in particular that they could be different, in Remark \ref{rmk:SE_mSE_diff}.
\begin{lemma}\label{lemma:CSE_wedge}
Assume the conditions in Assumption \ref{assumption:1}.
    Define
    $$\mSE =[\ew(V_1\cup V_2) \setminus (\ew(V_1)\cup \ew(V_2))] \, \bigcap \, [\ew(W_1\cup W_2) \setminus (\ew(W_1)\cup \ew(W_2))].$$
    Then $\mSE$ has positive measure
    if both $\ew(V_1\cup V_2)$ and $\ew(W_1\cup W_2)$ are connected.
\end{lemma}

\begin{proof}
    
Specifically, we show that when both $\ew(V_1\cup V_2)$ and $\ew(W_1\cup W_2)$ are connected
\begin{align}\label{eq:SE_modified}
    \ew(V_1)\cap T &=T_1, \quad \ew(V_2)\cap T=T_2, \nonumber \\
    \ew(W_1)\cap R &=R_1, \quad \ew(W_2)\cap R=R_2,  
\end{align}
where the ridge $T$ is the intersection between between future-pointing, inward-pointing null sheets from $RT(X_1)$ and $RT(X_2)$ and the ridge $R$ is the intersection between between past-pointing, inward-pointing null sheets from $RT(Y_1)$ and $RT(Y_2)$. Since the proof of Theorem \ref{thm:SE_wedge} establishes that $\mathcal{S}_E$ contains interior points of the ridges $T$ and $R$, and equation \eqref{eq:SE_modified} ensures these points lie outside $\ew(V_1)\cup\ew(V_2)\cup\ew(W_1)\cup\ew(W_2)$, it follows that $\tilde{\mathcal{S}}_E \neq \emptyset$.

First note that $\ew(V_1) \, \cap \, \bb = \hat{D}(V_1)=V_1$ and hence $T_1 \in \ew(V_1)$.
Assume by contradiction that another distinct point $p\in T$ also lies in $\ew(V_1)$. Since the ridge $T$ is the intersection between future-pointing, inward-pointing null sheets from $RT(X_1)$ and $RT(X_2)$, there exists a future-pointing null geodesic $\gamma_i$ connecting a point in the interior of $RT(X_i)$ to $p$ (note the unique null geodesic generator through $\partial RT(X_1)$ connects to $T_1$). Since $p\in \ew(V_1)$, then $\gamma_i$ is a null geodesic from the interior of $RT(X_i)$ to $\ew(V_1)$. Since $\ew(V_1)$ is defined as a domain of dependence, $\gamma_i$ should instead intersect the homology region of $V_1$. This contradicts with the fact that $V_1$ and $X_i$ are spacelike-separated, disjoint (sharing a boundary point) regions and hence have spacelike separated homology regions \cite{wall2014maximin}. Similar arguments establish other claims in \eqref{eq:SE_modified}.

\end{proof}


\begin{theorem}\label{thm:GCWT}
    Under the conditions of Assumption \ref{assumption:1}, 
    $\mSE$ has positive measure if and only if both $\ew(V_1\cup V_2)$ and $\ew(W_1\cup W_2)$ are connected. 
\end{theorem}
\begin{proof}
    Note that if either $\ew(V_1\cup V_2)$ or $\ew(W_1\cup W_2)$ is disconnected, then $\mSE$ is empty by definition. If either $\ew(V_1\cup V_2)$ or $\ew(W_1\cup W_2)$ is marginal, we define $\ew(V_1\cup V_2)=\ew(V_1)\cup \ew(V_2)$ and hence $\mSE =\emptyset$. In short, if either of  $\ew(V_1\cup V_2)$ and $\ew(W_1\cup W_2)$ is not connected, we have $\mSE=\emptyset$ by definition. Then the theorem follows from Lemma \ref{lemma:CSE_wedge}.
\end{proof}

\begin{corollary}\label{cor:SE_CSE}
    Assume the condition in Assumption \ref{assumption:1}. Then $$S_0\subseteq \mSE.$$
\end{corollary}
\begin{proof}
This follows directly from the original CWT and Theorem \ref{thm:GCWT}.
\end{proof}

\begin{remark}
   We note that explicit examples are given in \cite{lima2025sufficientGCWT} to show that $\mathcal{S}_E$ can either be empty or non-empty when either of $\ew(V_1\cup V_2)$ and $\ew(W_1\cup W_2)$ is disconnected. The current version of the generalized Connected Wedge Theorem, i.e. Theorem \ref{thm:GCWT}, puts more emphasis  than the proposed conjecture, i.e. Conjecture \ref{conj:LL}, on additional operators in $\mathcal{M}_{\ew(V_1\cup V_2)}$ and $\mathcal{M}_{\ew(W_1\cup W_2)}$, or equivalently, the strict superadditivity of boundary algebras.
\end{remark}
\begin{figure}
    \centering
    \includegraphics[width=\linewidth,trim={5 3.5cm 5 2cm},clip]{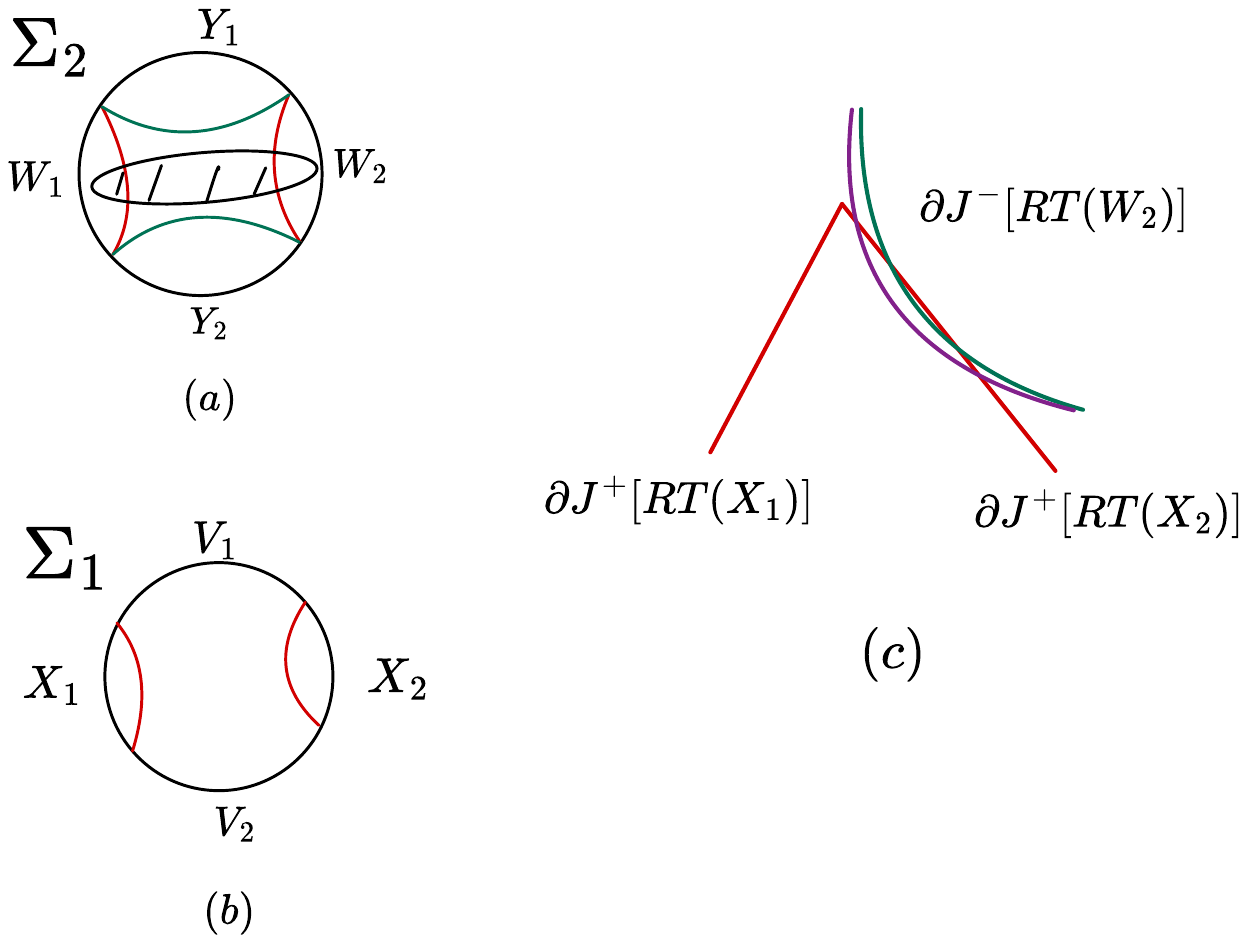}
    \caption{Schematic illustration of possible scenarios with $\tilde{\mathcal{S}}_E \subsetneq \mathcal{S}_E$. Panel $(a)$ shows $\Sigma_2$ along with HRRT surfaces of $W_1$, $Y_1$, $W_2$, and $Y_2$. The shaded region indicates where matter has been added to make $\ew(W_1\cup W_2)$ connected. Panel $(b)$ shows $\Sigma_1$ together with $RT(V_1\cup V_2) = RT(X_1) \cup RT(X_2)$. Panel $(c)$ shows a cross-section through the middle of $M$. The red curves delineate $\ew(V_1\cup V_2)$, while the green and purple curves show the past horizon of $\ew(W_1)$ before and after the addition of matter, respectively.}
    \label{fig:SE_mSE}
\end{figure}
\begin{remark}\label{rmk:SE_mSE_diff}
One naturally asks whether $\tilde{\mathcal{S}}_E = \mathcal{S}_E$ holds when both $\ew(V_1\cup V_2)$ and $\ew(W_1\cup W_2)$ are connected. We can modify examples from \cite{lima2025sufficientGCWT} to construct explicit configurations where $\tilde{\mathcal{S}}_E \subsetneq \mathcal{S}_E$. Consider a scenario where $\ew(V_1 \cup V_2)$ is connected and $\mathcal{S}_E \neq \emptyset$, but $\ew(W_1\cup W_2)$ is disconnected. This implies that either $\ew(W_1)$ or $\ew(W_2)$ has nontrivial intersection with $\ew(V_1\cup V_2)$. Assuming $\ew(W_2) \cap \ew(V_1\cup V_2) \neq \emptyset$, we add matter near $\Sigma_2$ such that $\ew(V_1\cup V_2)$ remains connected while $\ew(W_1 \cup W_2)$ becomes connected
(the metric deformation increases the separation along the $r_1$-$r_2$ direction).
The added matter retards past-directed and outward-directed null geodesics emanating from $RT(W_2)$. Consequently, the past horizon of $\ew(W_2)$ extends deeper into $\ew(V_1\cup V_2)$ (Figure \ref{fig:SE_mSE}). Since $\tilde{\mathcal{S}}_E$ excludes $\ew(W_2) \cap \ew(V_1\cup V_2)$ by definition, it becomes strictly smaller than $\mathcal{S}_E$.
\end{remark}

\section{Discussion}\label{section:discussion}
Similar to the proof of CWT, the proof given here will also work for semiclassical spacetimes that satisfy the quantum maximin formula \cite{akers2020quantummaxmin} and the quantum focusing conjecture \cite{bousso2016quantumfocusing}.

The original motivation of CWT is to give an operational restatement of the holographic principle -- ``An asymptotic quantum task is possible in the bulk if and only if the corresponding boundary task is"\cite{may2019quantumtask}.
Since the CWT concerns the situation where the bulk is described by low energy effective field theory, one would expect the CWT only goes in one direction: given an asymptotic quantum task that can be completed in the bulk using the low energy dynamics, one asserts that there must be a boundary procedure for completing the same task.
Also noted by May, the mutual information being positive does not imply the B84$^n$ task can be completed.

Given such considerations, it may appear surprising to have an equivalence statement, as proved here, between certain bulk classical geometry feature and a mutual information assertion on the boundary. However, the condition of connected entanglement wedge does not just imply existence of mutual information between decision regions but implies $O(1/G_N)$ mutual information. One is led to the following questions
\begin{itemize}
    \item Does $O(1/G_N)$ mutual information necessarily imply that the task must be performable on the boundary?
    \item Does $\mSE\neq \emptyset$ incorporate bulk quantum task somehow, given that $\mathcal{S}_0$ could be empty?
\end{itemize}
Understanding these questions from quantum information perspective would be of great conceptual importance.

Nevertheless, the proof here gives strong support that region $\mSE$ can be viewed as the bulk geometrization of the collection of all non-local quantum tasks that can be performed. In other words, any non-local quantum task has to be implemented via some operation in $\mSE$, and vice versa. Combined with the subregion-subalgebra duality, which states that $\mSE$ is dual to some boundary algebra $\mathcal{M}_{\mSE}$, the generalized connected wedge theorem  proved here (in principle) gives an explicit boundary theory characterization of all non-local quantum tasks \footnote{We thank Hong Liu for pointing this out to us.}.

Lastly, we note that ref. \cite{may2022nton} generalized the original CWT to $n$-to-$n$ process, which replaces the nonempty bulk scattering region by connectedness condition of a $2$-to-all causal graph. It would be of interest to promote this $n$-to-$n$ Connected Wedge Theorem to an equivalence statement. We leave this to future work.

\acknowledgments
I am deeply grateful to Edward Witten for his invaluable advice, stimulating discussions, and critical reading of the manuscript. I  thank Caroline Lima, Sabrina Pasterski and Chris Waddell for sharing their manuscripts and insights on the same problem. I also thank Sam Leutheusser, Hong Liu and an anonymous reviewer for their careful reading of the draft and helpful feedback. 
I acknowledge the hospitality of the Institute for Advanced Study (IAS), where the majority of this work was completed.

\appendix
\section{Some basic facts}\label{sec:appendix}

\renewcommand{\thesection}{\Alph{section}} 
\setcounter{Counter}{1}
\renewcommand{\thetheorem}{\Alph{section}.\arabic{theorem}} 

Here we collect some basic facts and their proofs which are invoked in the main text.

\begin{lemma}\label{lemma:A1}
    Consider a globally hyperbolic spacetime $M$. Let $S_1$ be a closed spacelike subset of $M$, which is contained in a Cauchy slice $\Sigma_1$. Consider $\partial J^+[S_1]$ and a cross section to the future, $S_2 = \partial J^+[S_1] \cap \Sigma_2$. Then we have the following:
    \begin{enumerate}
        \item To the future of $\Sigma_2$, $\partial J^+[S_2]$ coincides with $\partial J^+[S_1]$.
        \item Within $\Sigma_1$ and $\Sigma_2$, $\partial J^-[S_2]$ does not necessarily coincide with $\partial J^+[S_1]$. When they are distinct, $\partial J^-[S_2]$ lies to the past of $\partial J^+[S_1]$.
    \end{enumerate}
    The same conclusion holds for past causal boundaries.
\end{lemma}

\begin{proof}

    \begin{enumerate}
        \item Since $S_2 \subseteq J^+[S_1]$, $\partial J^+[S_2]\subseteq J^+[S_1]$. Let $p$ be an arbitrary point on $\partial J^+[S_1]$ to the future of $\Sigma_2$, i.e. $p \in \partial J^+[S_1]\cap J^+[\Sigma_2]$ (Figure \ref{fig:appfig1}(a)). Then there is a prompt (i.e. achronal) null geodesic $\gamma_{O\to p}$ that connects a point $O\in S_1$ and $p$. Since $\Sigma_2$ is a Cauchy slice, $\gamma_{O\to p}$ intersects $\Sigma_2$. 
        On the other hand, $S_2=J^+[{S_1}]\cap \Sigma_2$. Thus, $\gamma_{O\to p}$ intersects $S_2$, say at $q$. 
        Since $\gamma_{O\to p}$ is prompt or achronal, its restriction to $J^+[\Sigma_2]$, i.e. $\gamma_{q\to p}$ is also prompt. That is, $p \in J^+[S_2]$. 
        We thus proved that to the future of $\Sigma_2$, $\partial J^+[S_1]$ coincides with $\partial J^+[S_2]$. Reversing the time direction, one gets the same conclusion by replacing $J^+$ with $J^-$.
        \item Any point $q \in S_2$ is connected to a point $O \in S_1$ by a prompt null geodesic $\gamma_{O\to q}$. Reserving the direction $\gamma_{q\to O}$ also lies in $\partial J^-[S_2]$. However, not every point on $S_1$ lies on a prompt null geodesic reaching $S_2$, due to infinitesimal focusing or colliding null generators. That causal wedge does not necessarily coincide with entanglement wedge is an example of this phenomenon. 
        Figure \ref{fig:focus_Minkowski} is another example. Let $S_1$ be the wedge on $\{t=0\}$ plane and $S_2$ be $\partial J^+[S_1]\cap {t=t_0}$. Considering $\partial J^-[S_2]\cap \{t=0\}$, instead of the sharp corner of $S_1$,  one sees a subset of the light cone of the vertex of $S_2$ (Figure \ref{fig:appfig1}(b)).

         When they are distinct, it is easy to see that $\partial J^-[S_2]$ lies to the past of $\partial J^+[S_1]$. Since if otherwise, we would get contradiction with $S_2\subseteq J^+[S_1]$.
    \end{enumerate}
    \begin{figure}
        \centering
        \includegraphics[width=0.8\linewidth,trim={5 2cm 5 5},clip]{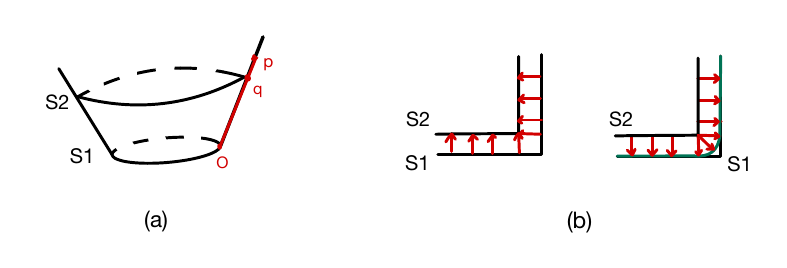}
        \caption{Illustration of relation between $\partial J^+[S_1]$ and $\partial J^\pm[S_2]$. $(a)$ $\partial J^+[S_1]$ and $\partial J^+[S_2]$ coincide in the future of $S_2$. $(b)$  $\partial J^+[S_1]$ and $\partial J^-[S_2]$ does not necessarily coincide in-between $S_1$ and $S_2$.}
        \label{fig:appfig1}
    \end{figure}
\end{proof}

\begin{figure}
    \centering
    \includegraphics[width=0.5\linewidth]{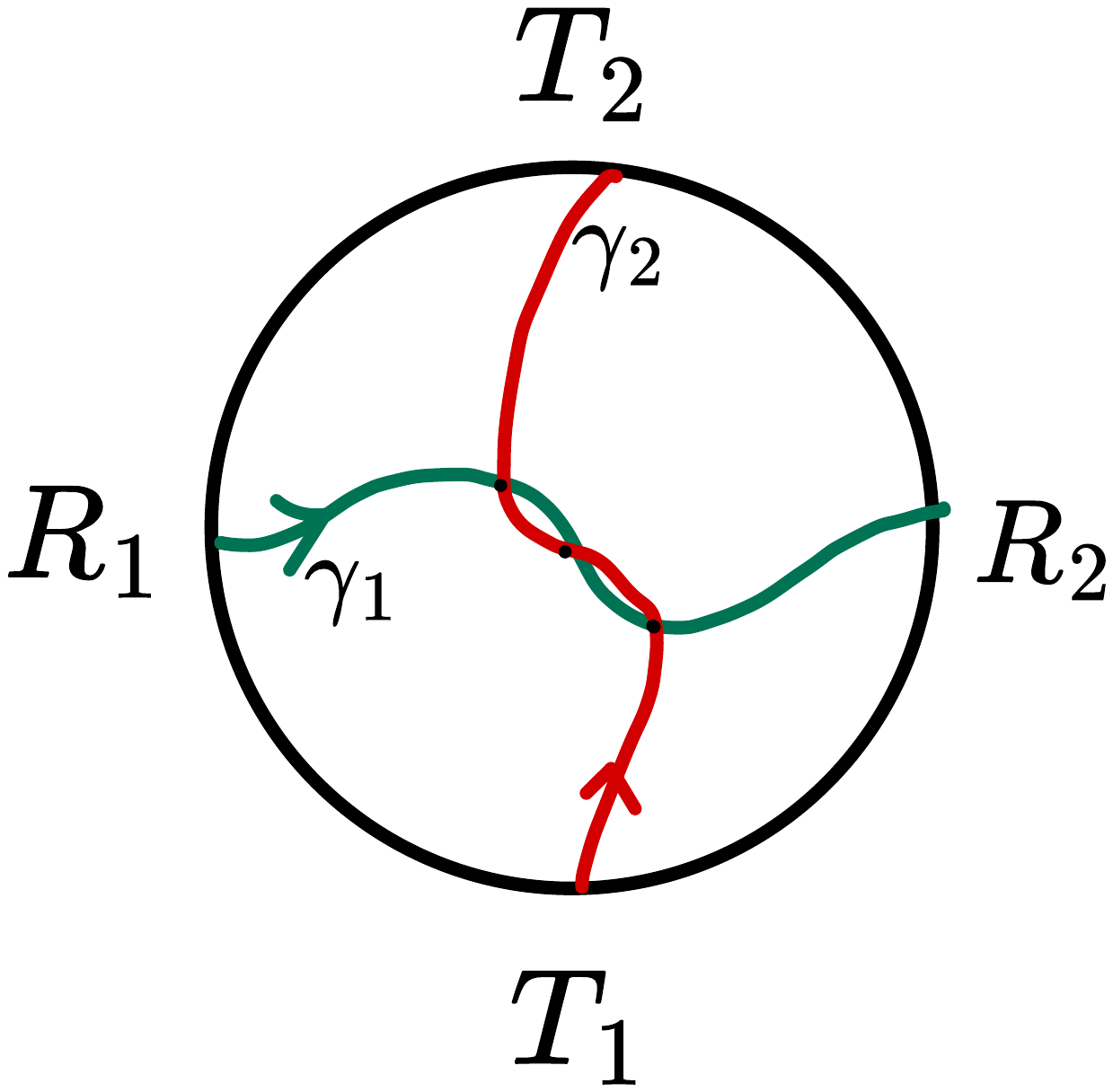}
    \caption{Illustration of signed intersection number of two curves in a disk, with interleaved endpoints. The curve $\gamma_1$ has endpoints $R_1$ and $R_2$ while the curve $\gamma_2$ has endpoints $T_1$ and $T_2$. Each curve is assigned an orientation. At each intersection point, a $\pm$ sign is assigned based on the relative orientation of the tangent vectors of $\gamma_1$ and $\gamma_2$.}
    \label{fig:signed_intersection}
\end{figure}
\begin{lemma}\label{lemma:A2}
    Let $\gamma_1$ and $\gamma_2$ be two simple (non-self-intersecting) curves in a simple disk $\mathbb{D}$. Denote the endpoints of $\gamma_1$ by $R_1$ and $R_2$, i.e. $\gamma_1 \cap \partial \mathbb{D} =\{R_1, R_2\}$. Denote the endpoints of $\gamma_2$ by $T_1$ and $T_2$, i.e. $\gamma_2 \cap \partial \mathbb{D} =\{T_1, T_2\}$. If $R_1, T_1, R_2, T_2$ are of cyclic order along $\partial\mathbb{D}$, then the intersection number of $\gamma_1$ and $\gamma_2$ is an odd integer.
\end{lemma}
\begin{proof}
Here we give an intuitive explanation. One can use the signed intersection number, a concept from algebraic topology, to derive this fact rigorously.

Since $\gamma_2$ is a simple curve (with no self-intersections) and has endpoints on the boundary $\partial \mathbb{D}$, it divides the disk $\mathbb{D}$ into two regions. The cyclic order of the endpoints $R_1, T_1, R_2, T_2$ implies that $\gamma_1$ must start in one region and end in the other. As $\gamma_1$ traverses the disk, it must cross $\gamma_2$ an odd number of times, even if it moves back and forth. The reader may convince themselves that this necessarily results in an odd integer intersection number.
\end{proof}

\begin{figure}
    \centering
    \includegraphics[width=0.6\linewidth]{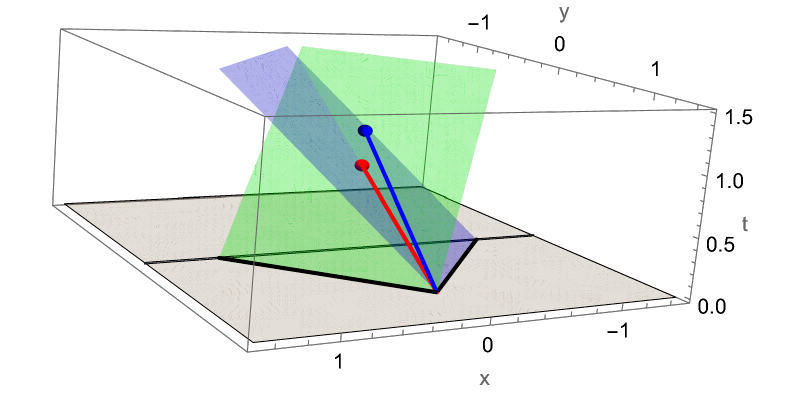}
    \caption{An elementary example in the Minkowski space-time analogous to $\ew(A)\cap \bb=\hat{D}(A)$. The null ray emitted from the vertex (blue) lies strictly to the future of the domain of dependence of the wedge, due to collision of null geodesics emitted from the two edges of the wedge.}
    \label{fig:focus_Minkowski}
\end{figure}
\begin{remark}
    It is well known that the entanglement wedge $\ew(A)$ always contains the causal wedge $\cw(A)$, although their intersections with the timelike boundary $\bb$ are both $\hat{D}(A)$. 
    An elementary example in the Minkowski spacetime may serve as an illustration.

    Consider a wedge in the $t=0$ plane of $R^{2,1}$, as shown in Figure \ref{fig:focus_Minkowski}. The cross-over seam, where light rays emitted from the two edges intersect, is at an angle less than $45^\circ$ with the $t=0$ plane. Correspondingly, the null ray emitted from the vertex toward the $y=0$ plane intersects the $y=0$ plane strictly to the future of the wedge's domain of dependence.
\end{remark}



\bibliographystyle{JHEP}
\bibliography{biblio.bib}


\end{document}